\documentclass[review,3p,10pt]{elsarticle}

\usepackage{graphicx,subfigure} 

\usepackage{amssymb}
\usepackage{amsthm}
\usepackage{amsmath}
\usepackage{array}
\usepackage{physics}
\usepackage{float}				
\usepackage{morefloats}
\usepackage{booktabs}
\usepackage{multirow}
\usepackage[linesnumbered,ruled]{algorithm2e}
\usepackage{enumitem}
\usepackage{arydshln}

\usepackage{mathtools}
\usepackage{stmaryrd}

\biboptions{sort&compress}

\mathchardef\Re="023C
\mathchardef\Im="023D

\usepackage{color}

        \newcommand{\CC}{{\mathcal C}}
        \newcommand{\DD}{{\mathcal D}}
        \newcommand{\FF}{{\mathcal F}}
        \newcommand{\AAA}{{\mathcal A}}
        \newcommand{\BBB}{{\mathcal B}}
        \newcommand{\HH}{{\mathcal H}}

        \newcommand{\KK}{{\mathcal K}}
        
        \newcommand{\MM}{{\mathcal M}}
        \newcommand{\OO}{{\mathcal O}}
        
        \newcommand{\XX}{{\mathcal X}}
        \newcommand{\YY}{{\mathcal Y}}
        
        \newcommand{\VV}{{\mathcal V}}

        \newcommand{\RRR}{{\mathcal R}}

				\newcommand{\CCC}{\ensuremath{\mathbb{C}}}
				\newcommand{\RR}{\ensuremath{\mathbb{R}}}
				\newcommand{\ZZZ}{\ensuremath{\mathbb{Z}}}

        \newcommand{\E}{{\rm E}}

        \newcommand{\diag}{{\rm diag}}

				\newcommand{\beq}{\begin{equation}}
				\newcommand{\enq}{\end{equation}}

        \newcommand{\beqa}{\begin{eqnarray}}
        \newcommand{\enqa}{\end{eqnarray}}
        \newcommand{\beqas}{\begin{eqnarray*}}
        \newcommand{\enqas}{\end{eqnarray*}}
        \newcommand{\bea}{\begin{array}}
        \newcommand{\ena}{\end{array}}
        \newcommand{\vect}[1]{\begin{bmatrix}#1\end{bmatrix}}
        \newcommand{\eqdef}{\stackrel{\rm def}{=}}

				\newcommand{\est}[1]{\hat#1}		

                \renewcommand{\H}{{\rm H}}  
                \newcommand{\vecc}{{\rm{vec}}}

\newtheorem{definition}{Definition}
\newtheorem{lemma}[definition]{Lemma}

\newtheorem{proposition}[definition]{Proposition}

\newtheorem{remark}[definition]{Remark}

\def\QEDopen{{\setlength{\fboxsep}{0pt}\setlength{\fboxrule}{0.2pt}\fbox{\rule[0pt]{0pt}{1.3ex}\rule[0pt]{1.3ex}{0pt}}}}
\def\QED{\QEDopen} 

\def\proof{\noindent\hspace{2em}{\textbf{\itshape Proof:} }}
\def\endproof{\hspace*{\fill}~\QED\newline}

\usepackage{stackengine}

\DeclareFontFamily{U}{mathx}{\hyphenchar\font45}
\DeclareFontShape{U}{mathx}{m}{n}{<-> mathx10}{}
\DeclareSymbolFont{mathx}{U}{mathx}{m}{n}
\DeclareMathAccent{\widebar}{0}{mathx}{"73}

\journal{Mechanical Systems and Signal Processing}

\begin{document}

\begin{frontmatter}

\title{Fault detection on manifolds of nonlinear dynamical systems with dual autoencoders}

\author{Bulut Ku{\c{s}}konmaz }
\author{Szymon Gre{\'s}\corref{cor1}}
\author{Rafa{\l} Wi{\'s}niewski}

\cortext[cor1]{Corresponding author; {\it E-mail address: sgres@es.aau.dk }}

\address{Aalborg University, Department of Electronic Systems, Fredrik Bajers Vej 7C, 9220 Aalborg, Denmark}

\begin{abstract}
Autoencoders are commonly used for unsupervised data-driven fault detection in nonlinear dynamical systems. Despite their widespread success and often favorable performance compared with traditional approaches, most applications rely on heuristic reconstruction of measured data using features learned from nominal training data, without explicit insight into the underlying nonlinear dynamics. This lack of interpretability limits the extension of autoencoder-based fault detection methods to higher levels of fault diagnosis, e.g., fault localization and quantification, and confines their use largely to application-oriented studies. To address this limitation, we propose a strategy for detecting parametric faults in nonlinear stochastic mechanical systems. A mathematical representation of the output data is developed using Koopman operator theory, which motivates their embedding on a manifold and its subsequent approximation with a two-stage autoencoder. Fault detection is formulated within a hypothesis-testing framework, in which new data are tested for consistency with a neighborhood of the manifold identified from nominal observations. The proposed method is validated through Monte Carlo simulations of a toy mechanical system with two types of nonlinearity and applied to two well-known real benchmarks, where it provides favorable fault-detection performance compared with standard autoencoders.

\end{abstract}


\end{frontmatter}


\section{Introduction}\label{sec:intro}

Nonlinear phenomena are ubiquitous in the vibration response of modern structural systems, where they may arise from contact, friction, geometric effects, material behaviour, gravity, damage-induced mechanisms, among other phenomena \cite{worden2019nonlinearity}. 
Methods for identification of nonlinear dynamics have recently gained popularity due to the massive advances in machine learning to perform nonlinear regression, also coupled with the boost in the computational performance of the current hardware; for a survey on methodology see \cite{Schoukens_nonliearIDsurvey}. Closely related to system identification is the field of fault detection, where system features identified from the process data and in some cases the prior physical knowledge, are used to define residuals and make decisions on the presence or absence of faults in a monitored system; for a detailed review on vibration-based fault detection methods see \cite{AVCI2021107077}. 
Accordingly, a common taxonomy of fault-detection approaches distinguishes between data-driven and physics-based methods.

Physics-based fault-detection methods can be further categorized based on the amount of physical information used to aid in the fault diagnostics, ranging from first-principles models to physics-enhanced data-driven methods; see, for example, surveys in \cite{Haywood_2024,Cicirello_2024}. Within this class, nonlinear state observers are commonly used to detect faults, e.g., by applying statistical tests to innovation or residual sequences \cite{Kadirkamanathan2002ParticleFilteringFD}, evaluating likelihoods or generalized likelihood ratios computed from particle-filter approximations \cite{Li2001ParticleFilteringLikelihoodRatio,Zhang2005BootstrapParticleFiltersFDI}, parametrizing stable observer-based residual generators with thresholding schemes \cite{Yang2016ParametrizationNonlinearObserverFD}, or using multiple filters associated with different fault hypotheses \cite{Tadic2014ParticleFilteringSensorFaultDiagnosis}.
While state observers offer a well-developed change detection framework, they require sufficiently accurate nominal model, including its structure, parameters, and uncertainty description, among others, which presents a significant engineering cost and is not always available in practice.
Alternatively, fault-related parameters of nonlinear physics-based models can be estimated jointly with the system states, which constitutes fault detection, localization and quantification in one shot, and is a part of nonlinear model updating field. Therein, the dedicated methods use, e.g.,  bifurcation curves \cite{MELOT2026113589}, hierarchical Bayesian model updating \cite{JIA2022114646} or rely on nonlinear Bayesian filtering schemes, e.g., extended Kalman filters \cite{Ebrahimian2015EKFMaterialParameterEstimation}, unscented Kalman filters \cite{Chatzi2009UKFPFNonlinearStructuralID,ASTROZA2019782,IMPRAIMAKIS2022108026}, marginalized filters with disturbance rejection \cite{ASWAL2025113239,KUNCHAM2023110269}, among many other approaches; reader interested in detailed description of Bayesian filers and their performance comparison in parameter estimation is referred to \cite{Tatsis_2022,ASTROZA2019520}.
These approaches often alleviate the need for a fully calibrated white-box model by inferring selected uncertain, or fault-related quantities from data, which however, is obtained at the cost of an increased computational burden, especially when many parameters must be estimated jointly with the states.
A relatively new class of physics-based methods are physics-informed neural networks (PINNs), in which governing equations and structural constraints are incorporated into the learning process \cite{RAISSI2019686,Karniadakis2021PhysicsInformedML}. PINNs have been used directly for fault detection and condition monitoring, for example by encoding prior fault characteristic of the model in a learning loss \cite{Shen2021PhysicsInformedBearingFD}, combining physics-informed autoencoders with sequential probability ratio testing for novelty detection \cite{Lai2024PIDAE}, and developing uncertainty-aware diagnostic networks \cite{Xu2024PIPDN,Herwig2025BearingNet}. While PINNs offer an interpretable framework for fault detection, their accuracy depends on the correctness of the assumed fault parameterization and standard formulations can incur substantial training costs.

Within the machine-learning community, fault detection is often formulated as a one-class classification problem with only nominal data available \cite{farrar2012structural}. Early work on fault detection in nonlinear mechanical systems used support vector machines to construct damage-sensitive decision boundaries from output data \cite{BORNN2010909}. Their performance, however, depends strongly on the choice of kernel and its hyperparameters, for which no generally optimal selection exists.
Another line of work is based on approximation of deterministic nonlinear systems by Koopman operator, where deep networks are used to learn the lifting coordinates and subsequently to generate residuals for fault detection and isolation \cite{BAKHTIARIDOUST2023200,Nader_deepKoopman}. 
The Koopman framework enables linear residual generation and decision making tools to be applied in a lifted space, but existing methods have largely considered deterministic Koopman approximations, whereas dynamical systems in practice are affected by stochastic
inputs and measurement noise. Other deep learning methods and in particular autoencoders (AEs) \cite{hinton2006reducing}, learn instead a nominal latent representation directly from fault-free operating data and detect faults through deviations in reconstruction, prediction, or orthogonal residuals \cite{Simpson_modelreductionnonlinearAEs,CACCIARELLI2022107853,WANG2024105804}. More recent deep learning architectures extend this principle by learning nonlinear fault-sensitive features directly from vibration responses under random excitation \cite{Joseph2024DeepLearningArchitectures}. While the latent variables generated by the learning process can be viewed as fault-sensitive features, they usually are without any explainable interpretation, e.g., with respect to the system dynamics. In consequence, the current deep learning methods for fault detection are application-oriented and there is a lack of a provable methodical framework for detection of parametric changes in nonlinear dynamical systems.

In this work we propose a deep learning framework designed to detect  faults in parameters of output data collected from nonlinear mechanical systems. The proposed method is inspired by two past contributions: first, the subspace-based fault detection, where the parametric changes in linear time-invariant systems are detected by monitoring 
an expected value of a residual of a subspace of an output covariance Hankel matrix \cite{DOHLER2014207,DOHLER20132734} and second, on learning time-delay embeddings of nonlinear systems using Koopman theory \cite{Brunton2017ChaosHAVOK}.
In our work, the main idea is the following: we train a two-stage (dual) AE where the training of the first stage aims at reproducing a set of delay output covariance Hankel matrices, while the second stage aims at learning directions complementary to the latent space obtained in the first training stage. 
The fault detection problem is then cast in a hypothesis testing framework where the output covariance Hankel matrix obtained from a new data is tested to be a point in the neighborhood of a manifold characterized in the first training stage, or not. The detection residual is obtained based on mappings from both training stages.
The decision about the damage is taken by comparing its norm to a threshold obtained from nominal data. 
Alongside these methodological developments, a mathematical model of the delay output covariance Hankel matrix justifying the manifold embedding and thus the proposed fault detection methodology is developed using the Koopman operator theory.
Our contributions are summarized as follows:
\begin{itemize}
\item we show that a delay output covariance Hankel matrix of a nonlinear system is parametrized by Koopman modes and system states,
\item a collection of covariance Hankel matrices obtained from different realizations of the stochastic system may be interpreted as samples from a local manifold, 
\item a dual AE is used to approximate the manifold and its complementary directions from a collection of nominal data,  
\item a fault detection strategy is proposed where Hankel matrices corresponding to a faulty nonlinear system are shown not to lie on the nominal manifold.
\end{itemize}

The proposed methods is validated on a numerical Monte Carlo simulation of a 6 degree of freedom chain system with nonlinear spring stiffness. Two types of spring nonlinearities are considered: a periodic one and a cubic one. The application consists of a study of two fault diagnosis benchmarks: the S101 bridge and Vestas V27 wind turbine blade. Both in the numerical simulation and in the application, the proposed approach is compared against state of art AEs for fault diagnosis, where it shows a superior performance. 

The remainder of the paper is organized as follows. Section~\ref{sec:prob_state} formulates the problem. Nonlinear modeling based on a Koopman embedding is introduced in Section~\ref{sec:koopman_NL}, while the Hankel matrix structure of the introduced model is presented in Section~\ref{sec:Hankelmanifolds}. Section~\ref{sec:AE_manifold} introduces the proposed autoencoder architecture to learn Hankel-matrix-based representations. The fault detection scheme introduced in Section~\ref{sec:faultdet} is validated with nonlinear dynamical system examples presented in Section~\ref{sec:validation}.  Further evaluation of the detection scheme for real system applications is introduced in Section~\ref{sec:application}, and conclusions are drawn in Section~\ref{sec:conc}.

\section{Problem statement}\label{sec:prob_state}

Consider a parametric discrete-time state-space model to represent the dynamics of the monitored mechanical system
\begin{align}
        x_{k+1} &= f(x_k,\theta,w_k), \label{eq:statenl} \\
        y_k &= h(x_k,\theta) + v_k, \label{eq:outputnl}
\end{align}
where  $x_k \in \RR^m$ denotes the system state vector, $y_k \in \RR^r$ is the output vector, $f(x_k,\theta,w_k)\in \RR^m$ with $h(x_k,\theta)\in \RR^r$ respectively represent nonlinear state transition and measurement maps, and $\theta \in \RR^d$ is fault-sensitive system parameter. The system order $m$ is considered to be known and $r$ is the number of the measured outputs. 
The process $w_k \in \RR^m$ and measurement $v_k \in \RR^r$ noises are assumed to be independent and identically distributed vectors with zero mean and positive definite covariances.
Note that in the above it is assumed that the monitored system is characterized by some physical parameters. The model parameter vector $\theta$ is an image of system parameters. In fact, the physical parameters, model parameter $\theta$ and the system model itself are assumed to be unknown in this work. 

The purpose of this work is to detect changes in system parameter from its nominal value $\theta_0$ based on measurements $y_k$ generated from the system under (unknown) $\theta$ and unknown noise. This is framed within a classical hypothesis testing framework, which is related to detecting changes of a parametric residual vector ${\zeta}_{\theta_0}$
\begin{align}
  &\mathrm{H}_0:\quad {\zeta}_{\theta_0} = 0 \Rightarrow  \theta = \theta_0 \label{eq:iH0}\\
  &\mathrm{H}_1:\quad   {\zeta}_{\theta_0} \neq 0 \Rightarrow \theta \neq \theta_0 \label{eq:iH1}.
\end{align}
The design of ${\zeta}_{\theta_0}$ suitable for detecting parametric changes in nonlinear systems traditionally relies on model-based methods. While purely data-driven machine-learning approaches, e.g., AEs, are also effective in practice in this context, they often provide limited interpretability in terms of the underlying system dynamics and, consequently, offer weak theoretical guarantees for fault detection.

In what follows, we establish a modelling framework where the Koopman operator is used as a tool to obtain a parametric features of a nonlinear system \eqref{eq:statenl}-\eqref{eq:outputnl}. We show that these features span a manifold whose characteristic and complementary representations can be approximated based on data by using a particular autoencoder architecture. Subsequently, both the characteristic manifold and complementary coordinates are used to design a residual which mimics the well-known subspace-based fault detection residual and, under some assumptions, is shown to satisfy hypotheses \eqref{eq:iH0}-\eqref{eq:iH1}.


\section{Background on nonlinear modelling}\label{sec:koopman_NL}

A Koopman embedding for parametric deterministic discrete-time systems is established for pedagogical reasons first, which is generalized to the case of stochastic systems later. Hereafter the symbol $\theta$ is dropped to avoid cluttering the notation in this section.

\subsection{Koopman representation of deterministic system}

For the sake of simplicity, consider an autonomous mechanical system 
\begin{align}
\label{eq:state}
x_{k+1} &= f(x_k)
\end{align}
where $f: \XX \rightarrow{} \RR^m $ is a map, $\XX\subseteq \RR^m$ is a state-space, $x_k\in\XX$ and $k\in\ZZZ$ is discrete time.
A measurement of the system \eqref{eq:state} is defined as 
\[
y_k = h(x_k),
\]
where 
$h: \XX \rightarrow{} \RR^r$, $h \in \FF$ and $\FF$  denotes the set of bounded measurable functions. 
The deterministic Koopman operator $\KK:\FF\rightarrow{}\FF$ is a linear, generally infinite-dimensional, composition operator acting on observables
\[
\KK h(x) = h(f(x)) = h \circ f(x),
\]
where $\circ$ denotes map composition. Assume that the spectrum of the Koopman operator is discrete and let $\lambda_j\in\CCC$ be the eigenvalues of the deterministic Koopman operator and $\phi_j:\XX\to\CCC$ the corresponding deterministic Koopman eigenfunctions satisfying the eigendecomposition $\phi_j(f(x)) = \lambda_j \phi_j(x)$. In general, an arbitrary observable $h\in\FF$ does not admit an expansion in Koopman eigenfunctions. Therefore, we assume that, for each fixed parameter value $\theta$, the measurement map $h$ admits a discrete Koopman spectral representation in terms of eigenfunctions $\{\phi_j\}_{j=1}^{\infty}$ of $\KK$, with corresponding discrete spectrum $\{\lambda_j\}_{j=1}^{\infty}$. This assumption is relevant for mechanical systems operating near an attracting hyperbolic equilibrium and within a bounded operating region, where a few damped or oscillatory components dominate. Thus, for the finite-dimensional construction used below, this representation is truncated to the first $n$ dominant eigenfunctions and eigenvalues. Under this assumption, the measured observable admits the decomposition \cite{Mauroy_Koopmanbook,BEVANDA2021197}
\[
h(x) = \sum_{j=1}^{n} v_j\phi_j(x) + \varepsilon_n(x),
\]
where $v_j\in\CCC^r$ are the deterministic Koopman modes associated with the measured output and $\varepsilon_n(x)$ is a residual from neglected eigenfunctions, or approximation errors. Hereafter it is assumed that, for a prescribed small tolerance $\bar{\varepsilon}>0$, $n$ is chosen sufficiently large such that $\|\varepsilon_n(x)\|_{\infty}\leq \bar{\varepsilon}$ on the region of interest and $h(x)$ is represented by the retained eigenmodes and eigenfunctions of Koopman operator up to this residual. 
Note that this finite-dimensional approximation is applicable to regimes in which the response can be characterized by closed orbits.
In this work, the monitored mechanical systems are assumed to operate in a bounded neighbourhood of a stable regime, where a finite number of Koopman components dominates the measured response. This type of finite spectral approximation is commonly used in Koopman-based reduced-order modeling, for example in systems that are exactly or locally linearizable on the region of interest, including periodic or quasiperiodic dynamics and nonlinear systems in a neighborhood of a hyperbolic equilibrium \cite{BackesDragicevic2022}. The Koopman mode decomposition can then be used to characterize the measured output as
\[
y_k = \sum_{j=1}^{n} v_j\phi_j(x_k)+\varepsilon_n(x_k).
\]

\subsection{Koopman representation of stochastic systems} \label{sec:stochasticKoopman}
 
In this section, we distinguish the random variables from their realizations by using upper-case and lower-case letters, respectively. Thus, $X_k,Y_k,W_k,V_k$ denote random variables corresponding to the state, measured output, process noise, and measurement noise, while $x_k,y_k,w_k,v_k$ denote their realizations. The parameter $\theta$ is treated as deterministic, and the noise-free measured observable is denoted by
$\bar{Y}_k \eqdef h(X_k)$ with $\bar{y}_k = h(x_k)$. With the process noise in \eqref{eq:statenl}, the state evolution
\[
X_{k+1}=f(X_k,W_k)
\]
is stochastic even for a fixed value of $X_k=x$ and $\theta$. Therefore, the Koopman operator associated with \eqref{eq:statenl} is not a deterministic composition operator, but a stochastic one. In this context, for an observable $g\in\FF$, the stochastic Koopman operator $\KK$ is defined as a conditional-expectation operator
\[
(\KK g)(x) \eqdef \E\left[g(X_{k+1})\mid X_k=x\right].
\]
In particular, for the measurement map $h$, using \eqref{eq:statenl} gives
\[
(\KK h)(x) = \E\left[h(f(x,W_k))\right],
\]
where the expectation is taken with respect to the distribution of the process noise $W_k$. Thus, $\KK$ remains linear on the observable space $\FF$, but it describes the evolution of conditional expectations, rather than sample path evolutions. Let $\lambda_j\in\CCC$ be the eigenvalues of the stochastic Koopman operator and let $\phi_j(x):\XX\to\CCC$ be the corresponding eigenfunctions. They satisfy
\[
(\KK \phi_j)(x) = \E\left[\phi_j(X_{k+1})\mid X_k=x\right]= \E\left[\phi_j(f(x,W_k))\right] = \lambda_j\phi_j(x).
\]
As in the deterministic case, we assume that the measurement map admits a discrete spectral representation, which is truncated to the first $n$ dominant components
\[
h(x) = \sum_{j=1}^{n}v_j\phi_j(x) + \varepsilon_n(x),
\]
where $v_j\in\CCC^r$ are the stochastic Koopman modes associated with the measured output and $\varepsilon_n(x)$ is the residual defined similarly as in the previous section.
The stochastic Koopman mode decomposition then characterizes the expected value of the noise-free measured response $\bar{Y}_k$ as
\[
\E\left[\bar{Y}_k\mid X_0=x_0\right] = \sum_{j=1}^{n} \lambda_j^k v_j \phi_j(x_0) + ((\KK)^k\varepsilon_n)(x_0),
\]
where
\[
((\KK)^k\varepsilon_n)(x_0) = \E\left[\varepsilon_n(X_k)\mid X_0=x_0\right].
\]
In the next step, let the system output be perturbed by the measurement noise $V_k$, $Y_k \eqdef \bar{Y}_k+V_k$. 
Since $V_k$ is zero mean and independent of the state process, the same expression also describes the conditional mean of the measured output
 with
\[
\E\left[Y_k\mid X_0=x_0\right] = \E\left[\bar{Y}_k\mid X_0=x_0\right],
\]
where for an individual realization of the output
\beq
\label{eq:ykstoch}
y_k = \sum_{j=1}^{n}v_j\phi_j(x_k) +\varepsilon_n(x_k) + v_k .
\enq

\section{Hankel matrix manifold of a nonlinear system}\label{sec:Hankelmanifolds}

In this section, first, a parametric link between the Koopman operator and the output covariances of the stochastic system is established and second, it is shown that the collection of the output data covariances are points on a local smooth manifold. 

\begin{definition}[Data matrix] \label{def:datamatrix_def} Let $m_k\in\RR^{b}$ be a discrete signal at time step $k$, $k = 0\hdots N+2p-1$ where $N+2p-1$ denotes the sample size and let the parameter $p$ define a data horizon. For $0\leq i\leq j \leq 2p - 1$ the data matrix $\MM_{i|j}$ writes
\begin{equation*}
\label{pastfutureY}
\MM_{i|j} \eqdef \frac{1}{\sqrt{N}}
\begin{bmatrix}
m_{i} & m_{i+1}& \ldots & m_{i+N-1}\\
m_{i+1} & m_{i+2}& \ldots & m_{i+N}\\
\vdots & \vdots & \vdots & \vdots\\
m_{j} & m_{j+1}& \ldots & m_{j+N-1}
\end{bmatrix}
\in \mathbb{R}^{(j-i+1)b \times N}.\!\!
\end{equation*}

\end{definition}

After the stochastic notation introduced in the previous section, $Y_k$ and $V_k$ denote random variables, while $y_k$ and $v_k$ denote their realizations. Let $\mathsf{Y}^- \eqdef \mathsf{Y}_{0|p-1}\in \RR^{pr \times N}$, $\mathsf{Y}^+ \eqdef \mathsf{Y}_{p\mid2p-1} \in \RR^{pr \times N}$ and $\mathsf{V}^- \eqdef \mathsf{V}_{0|p-1}\in \RR^{pr \times N}$, $\mathsf{V}^+ \eqdef \mathsf{V}_{p|2p-1} \in \RR^{pr \times N}$ denote the random data matrices obtained from the output and measurement-noise processes by Definition \ref{def:datamatrix_def}. Their realizations, obtained from samples $y_k$ and $v_k$, are denoted by ${\YY}^- \eqdef {\YY}_{0|p-1}$, ${\YY}^+ \eqdef {\YY}_{p\mid2p-1}$ and ${\VV}^- \eqdef {\VV}_{0|p-1}$, ${\VV}^+ \eqdef {\VV}_{p|2p-1}$, respectively. For example, ${\YY}_{0|p-1}$ is obtained from the realization $y_k$ and the $k$-th column of $\YY^-$ contains the past block $\vect{y_k^\top & y_{k+1}^\top& \hdots & y_{k+p-1}^\top}^\top$.
To characterize the parametric structure of data matrices, some notation is introduced first. Define $V(\theta) \eqdef \vect{v_1(\theta) & \hdots & v_n(\theta)} \in \CCC^{r\times n}$, $\Lambda(\theta) \eqdef \diag(\lambda_1(\theta),\hdots,\lambda_n(\theta)) \in \CCC^{n\times n}$ and 
\begin{align*}
   \phi(x_k,\theta) \eqdef \vect{\phi_1(x_k,\theta) \\ \vdots \\ \phi_n(x_k,\theta)} \in \CCC^{n} .
\end{align*}

Using Definition \ref{def:datamatrix_def} and output realization in \eqref{eq:ykstoch}, define
\begin{align*}
    \bar{\YY}^- \eqdef \vect{V(\theta)\phi(x_0,\theta) & V(\theta)\phi(x_1,\theta) & \hdots & V(\theta)\phi(x_{N-1},\theta) \\ V(\theta)\phi(x_1,\theta) & V(\theta)\phi(x_2,\theta) & \hdots & V(\theta)\phi(x_N,\theta) \\ \vdots & \vdots & \ddots & \vdots \\ V(\theta)\phi(x_{p-1},\theta) & V(\theta)\phi(x_p,\theta) & \hdots & V(\theta)\phi(x_{N+p-2},\theta)  }
\end{align*}
\begin{align*}
    \bar{\YY}^+ \eqdef \vect{V(\theta)\phi(x_p,\theta) & V(\theta)\phi(x_{p+1},\theta) & \hdots & V(\theta)\phi(x_{p+N-1},\theta) \\ V(\theta)\phi(x_{p+1},\theta) & V(\theta)\phi(x_{p+2},\theta) & \hdots & V(\theta)\phi(x_{N+p},\theta) \\ \vdots & \vdots & \ddots & \vdots \\ V(\theta)\phi(x_{2p-1},\theta) & V(\theta)\phi(x_{2p},\theta) & \hdots & V(\theta)\phi(x_{N+2p-2},\theta)}
\end{align*}
\begin{align*}
    {\RRR}^- \eqdef \vect{ \varepsilon_n(x_0) & \varepsilon_n(x_1) & \hdots & \varepsilon_n(x_{N-1}) \\ \varepsilon_n(x_1) & \varepsilon_n(x_2) & \hdots & \varepsilon_n(x_N) \\ \vdots & \vdots & \ddots & \vdots \\ \varepsilon_n(x_{p-1}) & \varepsilon_n(x_p) & \hdots & \varepsilon_n(x_{N+p-2})  }, \ \ 
    {\RRR}^+ \eqdef \vect{\varepsilon_n(x_p) & \varepsilon_n(x_{p+1}) & \hdots & \varepsilon_n(x_{p+N-1}) \\ \varepsilon_n(x_{p+1}) & \varepsilon_n(x_{p+2}) & \hdots & \varepsilon_n(x_{N+p}) \\ \vdots & \vdots & \ddots & \vdots \\ \varepsilon_n(x_{2p-1}) & \varepsilon_n(x_{2p}) & \hdots & \varepsilon_n(x_{N+2p-2})}
\end{align*}
The past ${\YY}^-$ and the future ${\YY}^+$ block Hankel matrices can be then written as ${\YY}^- = \bar{\YY}^- + \RRR^- + \VV^-$ and ${\YY}^+ = \bar{\YY}^+ + \RRR^+ + \VV^+$, where $\RRR^+$ with $\RRR^-$ are data matrices corresponding to the truncation error terms. 
To further explore the covariance structure of the realized data define  
\[
\OO(\theta) \eqdef \vect{V(\theta) \\ V(\theta) \Lambda(\theta) \\ \vdots \\ V(\theta)\Lambda(\theta)^{p-1}} \in \CCC^{pr\times n} .
\]

\begin{proposition}[Koopman embedding of output covariance Hankel matrix] \label{prop:hankel_covariance_koopman} Assume that the chosen stochastic Koopman eigenfunctions have bounded fourth moments and that for each fixed $\theta$ the Markov chain generated by \eqref{eq:statenl} is ergodic. Then, for a fixed finite realization, the empirical delay output covariance estimate admits the factorization 
\begin{align*}
\est{\HH}(x,\theta) &\eqdef \YY^+ {\YY^-}^\top  = {\OO}(\theta){\Gamma}(x,\theta)  + O(\bar{\varepsilon}) + O(\bar{\varepsilon}^2) + o(1) \\ &\approx \est{\OO}(\theta)\est{\Gamma}(x,\theta) ,
\end{align*}
where $(\est{\cdot})$ denote estimates of the exact quantities and
\[
\Gamma(x,\theta) \eqdef \vect{\Lambda(\theta)^p G(x,\theta){V}(\theta)^\H &\Lambda(\theta)^{p-1} G(x,\theta){V}(\theta)^\H &\hdots &
\Lambda(\theta) G(x,\theta){V}(\theta)^\H} \in \CCC^{n\times pr}, 
\]
with $G(x,\theta) \eqdef \frac{1}{{N}} \sum^{N-1}_{k=0} \phi(x_k,\theta) \phi(x_k,\theta)^\H$. 
\end{proposition}

\proof
See \ref{app:proofHcovKop}.
\endproof

Each finite output realization gives one realization of the empirical covariance estimate $\est{\HH}(x,\theta)$. What remains to show is that a collection of vectorized $\est{\HH}(x,\theta)$ are points on a local smooth manifold parametrized by $\theta$ and by the empirical second-order Koopman features contained in $\est{\Gamma}(x,\theta)$. 
For this purpose, consider a vectorized estimate of the Hankel covariance matrix
\begin{align}
\label{eq:HankelVec}
\est{h}(x,\theta) \eqdef \vecc(\est{\HH}(x,\theta)) = \BBB(\theta) g(x,\theta),
\end{align}
where
\[
\BBB(\theta) \eqdef I_{pr} \otimes \est{\OO}(\theta) \in \CCC^{(pr)^2\times npr}, \qquad g(x,\theta) \eqdef \vecc(\est{\Gamma}(x,\theta)) \in \CCC^{npr},
\]
and $(\bar\cdot)$ denotes complex conjugate, $\vecc(\cdot)$ denotes the column-stacking vectorization operator.

\begin{proposition}[Local manifold structure of $\est{h}$]
\label{prop:hankel_covariance_manifold}
Let $\mathcal X_N$ denote a smooth manifold of finite state
realizations required to construct the data matrices in
Definition~\ref{def:datamatrix_def}. 
%
Assume that the chosen stochastic Koopman eigenfunctions are smooth on the region of interest, and that the map
$x\mapsto\est{h}(x,\theta)$ in \eqref{eq:HankelVec} has the same rank in a neighborhood of each considered realization. Then any finite collection of vectorized covariance estimates $\{\est{h}^{(j)}(x,\theta)\}_{j=1}^J$, each obtained from one finite realization of the stochastic system through \eqref{eq:HankelVec}, may be interpreted as samples from a neighborhood of a smooth embedded manifold
\[
\mathcal M_H(\theta) = \left\{\est{h}(x,\theta):x\in\mathcal X_N\right\}\subset\mathbb R^{(pr)^2}.
\]
\end{proposition}

\proof
See \ref{app:proof_hankel_covariance_manifold}.
\endproof


\section{Manifold learning}\label{sec:AE_manifold}

Hereafter we adapt a two-stage autoencoder architecture to identify a low-dimensional latent representation that parametrizes the neighborhood of the manifold $\mathcal M_H(\theta)$ characterized in Proposition \ref{prop:hankel_covariance_manifold} and additionally, capture directions complementary to the characteristic manifold. In what follows, we introduce the necessary background for learning these representations. 

A typical autoencoder contains an encoder and a decoder neural network \citep{hinton2006reducing}. The encoder is used to map the data into lower dimensional latent representation while the decoder reconstructs the original data from the mapped latent representation with a small margin of error. Hereafter, we adopt a particular autoencoder training scheme inspired by \cite{shen2025analytical}, where the characteristic local manifold is identified based on samples of $\est{h}(x,\theta)$, alongside its normal bundle, i.e., directions complementary to the latent space of an autoencoder. For simplicity, we drop the symbols denoting the parametrization with the system parameter $\theta$ and state trajectories $x$ from the notation in this section. 

Recall from Proposition \ref{prop:hankel_covariance_manifold} that $\est{h}^{(j)}$ are $j= 1\hdots J$ samples from a smooth manifold $\MM_H \subset \mathbb{R}^{(pr)^2}$, where $J$ denotes the number of training data sets. Recall that $p$ denotes the number of time delays in the output covariance Hankel matrix and $r$ is the number of output sensors and in applications to large scale mechanical systems $(pr)^2$ can be large. The intrinsic dimension of $\MM_H$, however, is assumed to be much smaller than $(pr)^2$. The purpose of the first training is to learn an encoder map $\varphi:\RR^{(pr)^2}\to\RR^\ell$ and a decoder map $\psi:\RR^\ell\to\RR^{(pr)^2}$ such that each sample $\est{h}^{(j)}$ is mapped to a latent coordinate $z^{(j)}\in\RR^\ell$, where $\ell$ is the dimension of the latent space, and then reconstructed as $\check{h}^{(j)}$
\begin{align}\label{eq:AEMAPS}
\begin{array}{c}
\est{h}^{(j)} \xmapsto{\varphi} {z}^{(j)} \xmapsto{\psi} \check{h}^{(j)} \\
\{ \est{h}^{(1)}...\est{h}^{(J)}\} \longmapsto
\{z^{(1)}...z^{(J)}\} \longmapsto \{ \check{h}^{(1)}...\check{h}^{(J)}\},
\end{array}
\end{align}
which is achieved by minimizing an error between each $\est{h}^{(j)}$ and its corresponding reconstruction $\check{h}^{(j)}$  
\begin{align}
    \mathcal{L}_h=\sum_{j = 1}^{J}\|\est{h}^{(j)}-\check{h}^{(j)}\|.
    \label{eq:manifold_loss}
\end{align}
This process enables identifying the low-dimensional latent space of the manifold $\MM_H$. The autoencoder mappings used in the first training write as follows 
\begin{align*}
\est{h}^{(j)} \in \MM_H \subset \mathbb{R}^{(pr)^2} \xrightarrow{\varphi_1} \mathbb{R}^t \xrightarrow{\varphi_2} \mathbb{R}^\ell \ni z^{(j)} \xrightarrow{\psi_1} \mathbb{R}^t  \xrightarrow{\psi_2} \mathbb{R}^{(pr)^2} \ni \check{h}^{(j)}
\end{align*}
where $t$ is the dimension of the hidden layer and
\beq
\label{1st_train_map}
\varphi \eqdef \varphi_2\circ \varphi_1, \quad \psi \eqdef \psi_2\circ \psi_1
\enq
where for each training sample
\[
z^{(j)}=\varphi(\est{h}^{(j)}),
\qquad
\check h^{(j)}=\psi(z^{(j)}).
\]

In the second learning stage, the latent layer is augmented with complementary coordinates. The purpose of this stage is to learn  directions which are complementary to the characteristic latent representation of the manifold $\MM_H$. Let $\AAA\subset\RR^{(pr)^2}$ be an ambient hyper-rectangle containing the learned manifold, i.e., $\MM_H\subset \AAA$, and let $a^{(i)}\in\AAA$, $i=1\hdots N_a$, denote samples from this set. Subsequently, let $\tilde z^{(i)}\in\RR^{(pr)^2-\ell}$ denote the complementary latent variable, which represents the component of the ambient sample $a^{(i)}$ that is not captured by the characteristic latent coordinate $z^{(i)}\in\RR^\ell$ associated with the learned manifold.
The additional layer of the network is constructed such that the previous encoder maps $\varphi_1$ and $\varphi_2$ remain fixed, so that the characteristic coordinate $z^{(i)}=\varphi_2(\varphi_1(a^{(i)}))$ is preserved. A new complementary encoder map $\bar{\varphi}_2:\RR^t\to\RR^{(pr)^2-\ell}$ is then introduced to compute $\tilde z^{(i)}=\bar{\varphi}_2(\varphi_1(a^{(i)}))$. The decoder is also augmented by adding a complementary decoder contribution $\bar{\psi}_2:\RR^{(pr)^2-\ell}\to\RR^t$ so that the reconstructed ambient sample is obtained from both the characteristic and complementary coordinates. The second-stage mapping can be written as
\begin{align} \label{eq:AEMAPS2}
a^{(i)}\in\AAA\subset\RR^{(pr)^2}\xmapsto{\varphi_1}\RR^t\xmapsto{ \bar{\varphi}_2}\RR^{(pr)^2-\ell}\ni \tilde z^{(i)}\xmapsto{\bar{\psi}_2}\RR^t\xmapsto{\psi_2}\RR^{(pr)^2}\ni \check{a}^{(i)},
\end{align}
where
\begin{align}
\label{2nd_train_map}
z^{(i)} = \varphi_2(\varphi_1(a^{(i)})),\quad
\tilde z^{(i)} = \bar{\varphi}_2(\varphi_1(a^{(i)})),\quad
\check{a}^{(i)} = \psi_2\left(\psi_1(z^{(i)})+\bar{\psi}_2(\tilde z^{(i)})\right).
\end{align}
The mappings $\bar{\varphi}_2$ and $\bar{\psi}_2$ are optimized by minimizing the reconstruction loss
\[
\mathcal{L}_a=\sum_{i=1}^{N_a}\|a^{(i)}-\check{a}^{(i)}\|^2.
\]

The double autoencoder maps provide latent representations on which the change-detection residual is built in the next section. Each vectorized covariance Hankel matrix is treated as one point in the ambient data space, while the collection of nominal matrices obtained from different finite realizations locally forms the characteristic manifold $\MM_H$. This construction provides a nonlinear counterpart of subspace-based fault detection: in the linear case, nominal Hankel matrices are characterized by a fixed subspace and faults are detected through directions complementary to this subspace, whereas here the fixed subspace is replaced by a locally nonlinear manifold. The map $\varphi_1$ encodes the data into a high-dimensional latent representation, $\varphi_2$ extracts the part associated with $\MM_H$, and $\bar{\varphi}_2$ extracts the complementary coordinates describing local deviations from $\MM_H$. For points on the manifold, the complementary coordinate is expected to be null under idealized conditions, whereas points away from $\MM_H$ require nonzero complementary coordinates for accurate reconstruction. This representation is used to decide whether a covariance Hankel matrix is consistent with the manifold identified from nominal data. A graphical summary of the involved mappings $\{\varphi_1,\varphi_2,\psi_1,\psi_2\}$ and $\{\bar{\varphi}_2,\bar{\psi}_2\}$ is depicted in Figure~\ref{fig:mappings}.

\begin{figure}[t!] 
\center
\includegraphics[width=0.48\textwidth]{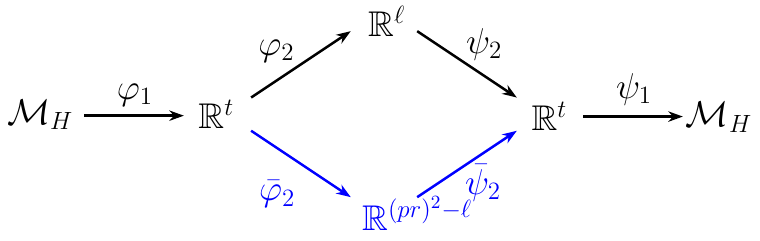}
\caption{Autoencoder mappings in the first and second training.} 
\label{fig:mappings}
\end{figure}

\section{Fault detection}\label{sec:faultdet}

For a fixed value of the nominal parameter $\theta_0$, variations of the nominal state trajectory are represented as variations along $\MM_H(\theta_0)$. In contrast, changes of $\theta$ from its baseline value are assumed to move the Hankel matrix away from $\MM_H(\theta_0)$, producing a nonzero transverse residual. Therefore, we assume that the manifolds $\MM_H(\theta_0)$ and $\MM_H(\theta)$ are locally separated. 

Recall that $\est{h}^{(j)}(\theta_0)$, $j= 1\hdots J$ denotes a collection of vectorized output covariance Hankel matrices obtained from the nominal data sets and let $\varphi_{1,\theta_0}$ and $\bar{\varphi}_{2,\theta_0}$ denote the corresponding autoencoder mappings. Consider a Hankel matrix obtained from data collected under an unknown value of system parameter $\est{h}(x,\theta)$. 
The decision about the change of $\theta$ can be achieved by monitoring deviations of $\est{h}(x,\theta)$ from the baseline manifold $\MM_H(\theta_0)$. To this end, we design a hypothesis testing framework based on a residual map $\zeta_{\theta_0} \eqdef  \bar{\varphi}_{2,\theta_0}\circ \varphi_{1,\theta_0}$, which is evaluated at each $\est{h}(x,\theta)$ as $\zeta_{\theta_0}(\est{h}(x,\theta)) =  \bar{\varphi}_{2,\theta_0}(\varphi_{1,\theta_0}(\est{h}(x,\theta)))$.

\begin{lemma}[Change detection hypotheses] \label{lemma:changehypotheses} Let $\DD \subset \RR^{(pr)^2}$ be the local domain on which the learned coordinates are considered with $\MM_H(\theta_0)\subset\mathcal D$. Assume an idealized autoencoder training, i.e., $\zeta_{\theta_0}^{-1}(0) \cap \DD = \MM_H(\theta_0)$. Then, for every $\est{h}(x,\theta)\in \DD$ it holds
\begin{align} 
  \mathrm{H}_0:\quad {\zeta}_{\theta_0}(\est{h}(x,\theta)) = 0  \Rightarrow \est{h}(x,\theta) \in \MM_H(\theta_0) \label{eq:H0}\\
  \mathrm{H}_1:\quad {\zeta}_{\theta_0}(\est{h}(x,\theta)) \neq 0 \Rightarrow \est{h}(x,\theta) \notin \MM_H(\theta_0). \label{eq:H1}
\end{align}

\end{lemma}

\proof
The proof follows directly from the assumed property $\zeta_{\theta_0}^{-1}(0)\cap \mathcal D = \MM_H(\theta_0)$ and is omitted here for brevity. 
\endproof

The lemma certifies membership of the tested Hankel $\est{h}(x,\theta)$ in the nominal manifold $\MM_H(\theta_0)$. To interpret non-membership as a parameter change, one additionally requires that deviations from $\MM_H(\theta_0)$ are attributed to changes in $\theta$, which is also assumed throughout the paper.

In practice, ${\zeta}_{\theta_0}(\est{h}(x,\theta_0))\approx 0$ due to a finite data length, measurement noise, approximation errors in the Koopman representation, and imperfect autoencoder training. Therefore, the decision between $\mathrm{H}_0$ and $\mathrm{H}_1$ is taken by comparing the norm of the residual to a threshold. For this purpose, let $\tau_{\theta_0}>0$ denote a threshold obtained on a nominal data and define $t \eqdef \left\| \zeta_{\theta_0}(\est{h}(x,\theta))\right\|_2$. Then, the test between the two hypotheses writes

\begin{align} \label{eq:finalhyptest}
\begin{array}{c}
  \mathrm{H}_0:\quad t < \tau_{\theta_0}  \\
  \mathrm{H}_1:\quad t \geq \tau_{\theta_0} . 
\end{array}
\end{align}

\begin{remark}
A robust hypothesis test should account for the uncertainty of the residual, including the effects of finite data, measurement noise, model approximation, and the uncertainty of the learned mappings and threshold. Developing such a statistical characterization is however beyond the scope of the current paper.
\end{remark} 

The algorithmic summary of the proposed approach is enclosed below.

\begin{algorithm}[H] \label{alg:imp}
	\caption{Fault detection on Hankel matrix manifolds.}  
	\SetKwInOut{Input}{Input}
    \SetKwInOut{Output}{Output}
    \Input{$J$ output training data sets and the output data for testing, number of time lags $p$, network parameters: hidden layer width, character space size $\ell$, learning rate, activation function type \;}
    \Output{test value $t$, decision about the fault\;}
    compute vectorized Hankel matrices $\{\est{h}^{(j)}(\theta_0)\}_{j=1}^J$ for $J$ training data sets after \eqref{eq:HankelVec} \;
    train autoencoder mappings $\{\varphi_{1,\theta_0},\varphi_{2,\theta_0},\psi_{1,\theta_0},\psi_{2,\theta_0}\}$ from the first training stage based on $\{\est{h}^{(j)}(\theta_0)\}_{j=1}^J$ after \eqref{eq:AEMAPS} and the chosen network parameters \;
    sample $\{a^{(i)}\}^{N_a}_{i=1}$ from a hyper-rectangle $\AAA$ around the nominal Hankel matrices \;
    train autoencoder mappings $\{\bar{\varphi}_{2,\theta_0},\bar{\psi}_{2,\theta_0}\}$ from the second training based on $\{a^{(i)}\}^{N_a}_{i=1}$ after \eqref{eq:AEMAPS2}  \;
    compute the threshold $\tau_{\theta_0}$ on $\{\est{h}^{(j)}(\theta_0)\}_{j=1}^J$, or preferably on other nominal calibration data \;
    for each vectorized Hankel matrix in a test state $\est{h}(\theta)$, compute $t$ and take decision based on $\tau_{\theta_0}$ \eqref{eq:finalhyptest}\;
\end{algorithm}


\section{Numerical validation}\label{sec:validation}

The approach proposed in Section \ref{sec:faultdet} is first validated on a numerical simulation of a nonlinear toy system. For this purpose, we consider a 6 degree of freedom (DOF) mechanical mass-spring system that (for any consistent set of units) is modeled with an initial spring stiffness $k_1 = k_3 = k_5 = 100$ and $k_2 = k_4 = k_6 = 200$, mass of each element $m_i = 1/20$ and a proportional damping matrix such that each mechanical mode has a damping ratio of $\zeta_i=3\%$. Two nonlinear analysis cases are considered; the first one where the spring nonlinearity is cubic and the second where it is a periodic function of systems' displacement. A detailed description of the nonlinear system and the data simulation is outlined in \ref{app:chainapp}. The system is excited with white noise inputs in all DOFs and sampled with a frequency of 500 Hz for 200 seconds. The acceleration outputs are collected at DOFs 1, 3 and 5. Standard Gaussian white noise with is added to the output sensors. 

\begin{figure}[ht!]
\centering
\includegraphics[width=0.75\linewidth]{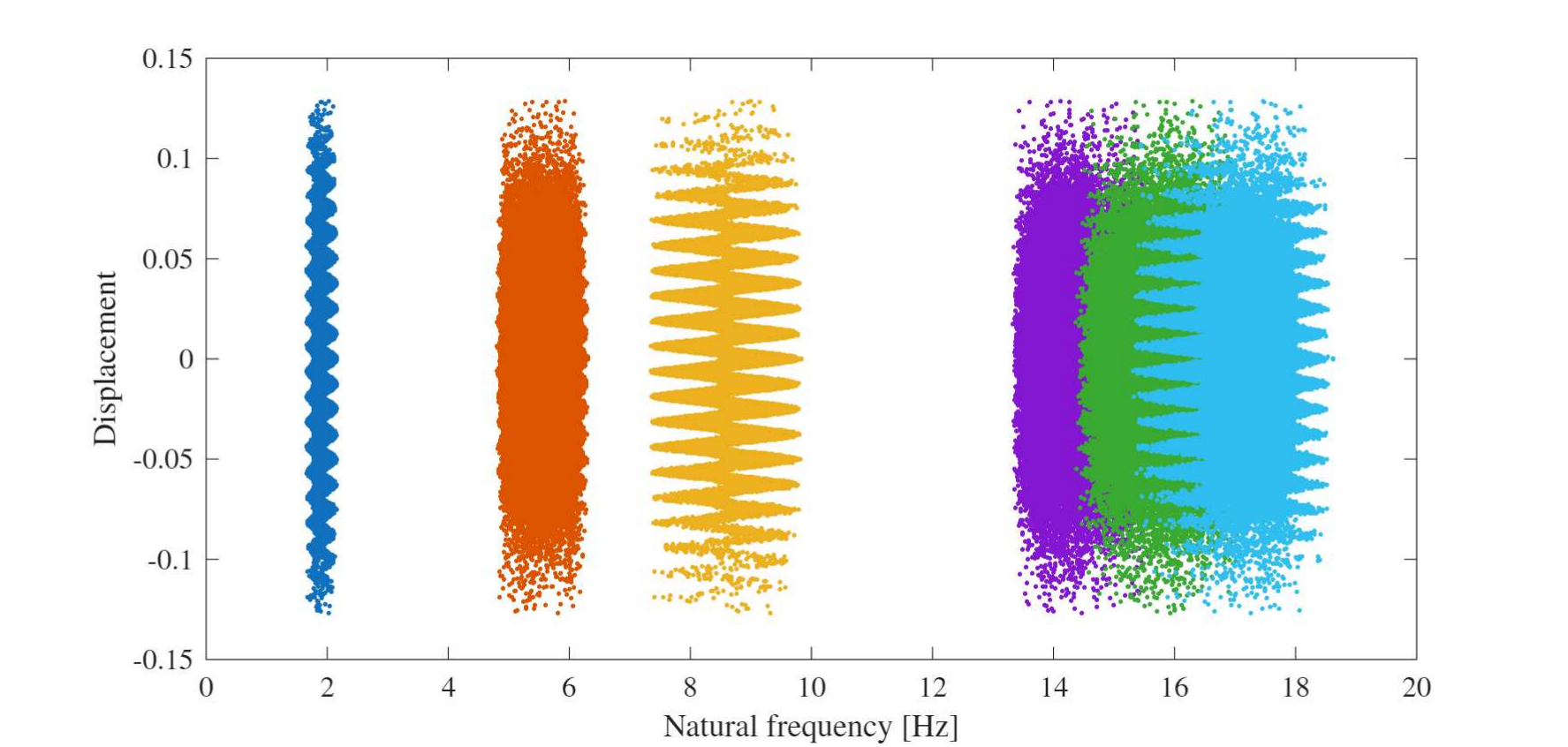}     
\caption{Displacement and the local natural frequency of the system with periodic nonlinearity.}
\label{fig:plotper}
\end{figure}

\begin{figure}[ht!]
\centering
\includegraphics[width=0.70\linewidth]{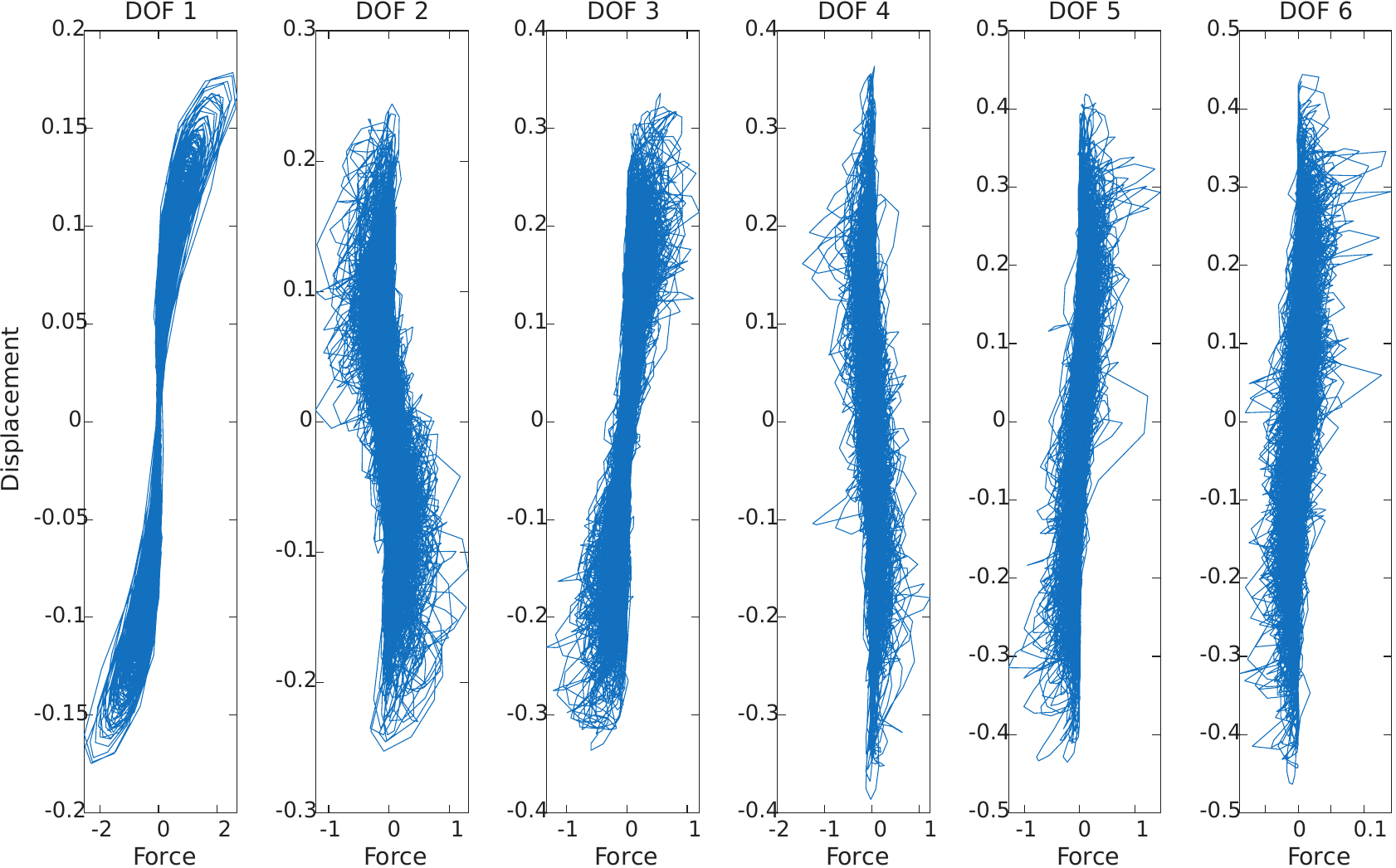}     
\caption{Displacement and restoring force of the system with cubic nonlinearity.}
\label{fig:plotcubic}
\end{figure}

Figure \ref{fig:plotper} illustrates a nonlinear relation between the natural frequency and displacement for the case of system with periodic stiffness and Figure \ref{fig:plotcubic} illustrates the relation between displacements and the restoring force for system with cubic nonlinearity. Both figures show that system parameters are a nonlinear functions of displacements.

The numerical Monte Carlo campaign consists of simulation of 100 data sets in a reference state of the system to estimate the baseline manifold, 1000 healthy data sets for the validation and 1000 data sets for each damage scenario. Two damage cases are considered; the first case where the fault is modeled as a gradual reduction in the stiffness of the second spring by 5\%, and the second where the spring stiffness is reduced by 10\%.

For all cases, the covariance Hankel matrices are computed with $p=15$ lags. The autoencoder uses two hidden layers of width $4096$ and a characteristic bottleneck dimension $\ell=20$. All hidden layers use a LeakyReLU activation with negative slope $0.1$, and the output layers are linear. The use of two hidden layers instead of one is an implementation choice intended to increase the approximation capacity of the encoder and decoder while keeping the architecture fully connected and symmetric. A single hidden layer network might lead to the same outcome, but may require a larger width to represent the same nonlinear transformation. 
Both training stages are optimized using a adaptive moment estimation (ADAM) with learning rate $10^{-4}$, mini-batch size $4$, and $1000$ training epochs. The model obtained after the final epoch is used for the computation of the residual scores. The weights converge with the loss of $0.022$ and $0.015$ respectively in the first and the second training stages with the computational time of 43\,s and 92\,s respectively on AMD Ryzen 7 5800H CPU with RTX3060 mobile GPU and 64GB of RAM. The corresponding algorithm was developed using PyTorch \cite{paszke2019pytorchimperativestylehighperformance}.
For both nonlinearity cases, the detection threshold is chosen heuristically as $95\%$ quantile of the empirical distribution of the test obtained on the reference data. 
The distributions of the test obtained for different damage levels and for two nonlinearity cases are shown in Figure \ref{fig:results_chain}. It can be viewed that the considered faults are well-detected for both nonlinearity types. The false positive rate is $7\%$ and $5\%$ respectively for the periodic and cubic nonlinearities and the false negative rate is $0\%$ and $5.3\%$  for the same cases. 
\begin{figure}[t!] 
\begin{minipage}[b]{0.49\linewidth}
\center
\includegraphics[width=0.98\textwidth]{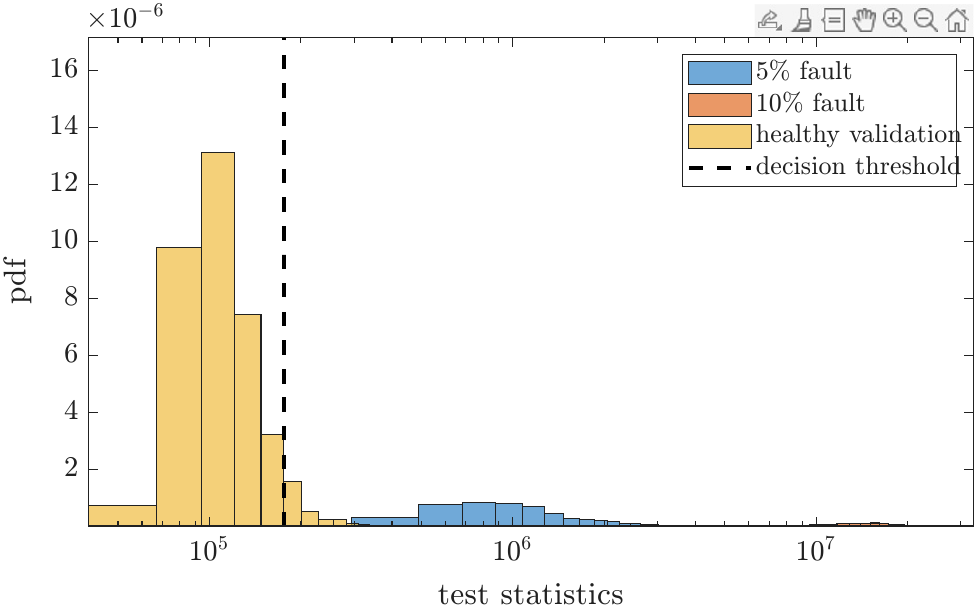}
\end{minipage}
\begin{minipage}[b]{0.49\linewidth}
\center
\includegraphics[width=0.98\textwidth]{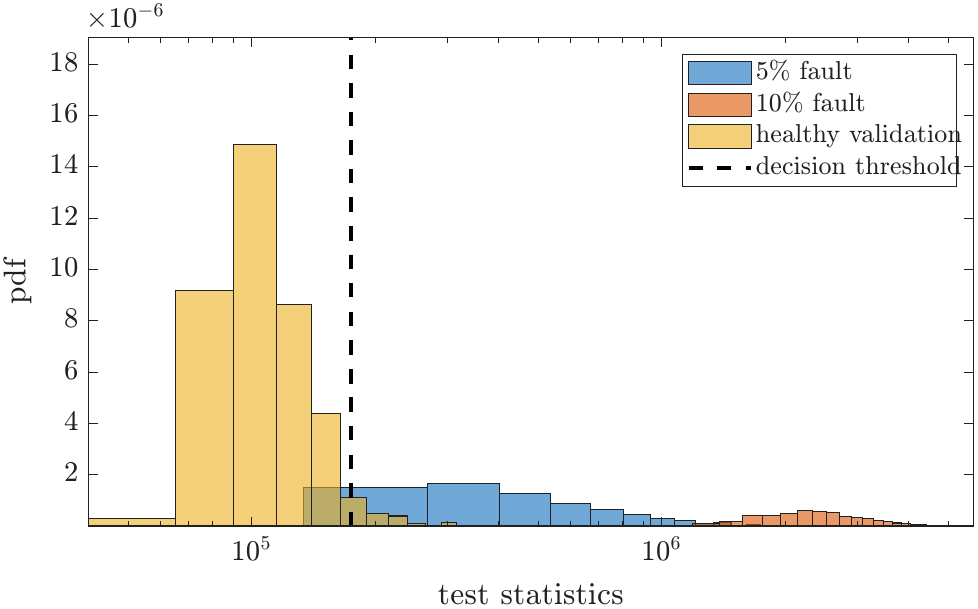}
\end{minipage}
\caption{Distributions of the residual norm for a system with periodic nonlinearity (left) and cubic nonlinearity (right).} 
\label{fig:results_chain}
\end{figure}
To showcase the performance of the proposed approach, we benchmark it against three autoencoder-based change detection methods: a fully connected autoencoder (AE), a variational autoencoder (VAE), and a beta-VAE (BVAE), implemented in state-of-the-art outlier detection package PyOD \citep{hinton2006reducing,kingma2013auto,zhao2019pyod,higgins2017beta}. The reconstruction error is used as the test statistic and the decision threshold is obtained based on the empirical 99\% quantile of the nominal test scores. All models were fitted only to the reference data sets. The variants denoted by AE 1, VAE 1 and BVAE 1 were trained on vectorized covariance Hankel matrices $\est{h}=\vecc(\est{\HH})$, while AE 2, VAE 2 and BVAE 2 were trained directly on the vectorized outputs. The data from the Monte Carlo simulation of a nonlinear chain system were used for the comparison. The weights were optimized with ADAM with the learning rate $10^{-5}$ and the hidden layer uses ReLU activation functions. The training parameters were $1000$ epochs, mini-batch size $32$, hidden layer size $1024$, bottleneck size $24$. The comparison with the proposed dual autoencoder, denoted as DAE, is illustrated in Table \ref{table:chain_valid}. The graphical illustration of the test results from AE 1 and VAE 1 obtained on data from a system with cubic nonlinearity is depicted in Figure \ref{fig:results_chainVAE}. 

\begin{table}[t!]
\caption{Comparison between the state-of-the-art autoencoder-based anomaly detection methods (AE, VAE, BVAE) and the proposed approach (DAE).}
\vspace{2mm}
\footnotesize
\centering
\begin{tabular}{ll*{7}{c}}
\toprule
{Nonlinearity} & {Metric} & {AE 1} & {AE 2} & {VAE 1} & {VAE 2} & {BVAE 1} & {BVAE 2} & {DAE} \\
\midrule
\multirow{2}{*}{Cubic}
& {False alarm rate} & {$12.5\%$} & {$68.4\%$} & {$2.1\%$} & {$50.4\%$} & {$2.0\%$} & {$50.2\%$} & {$5.0\%$} \\
& {Detection rate}   & {$100\%$}  & {$52.4\%$} & {$75.0\%$} & {$51.2\%$} & {$75.7\%$} & {$51.0\%$} & {$94.7\%$} \\
\midrule
\multirow{2}{*}{Periodic}
& {False alarm rate} & {$19.9\%$} & {$79.0\%$} & {$0.6\%$} & {$68.7\%$} & {$0.5\%$} & {$68.9\%$} & {$7.0\%$} \\
& {Detection rate}   & {$100\%$}  & {$69.0\%$} & {$81.0\%$} & {$52.4\%$} & {$81.5\%$} & {$54.2\%$} & {$100\%$} \\
\bottomrule
\end{tabular}
\label{table:chain_valid}
\end{table}

\begin{figure}[t!] 
\begin{minipage}[b]{0.49\linewidth}
\center
\includegraphics[width=0.98\textwidth]{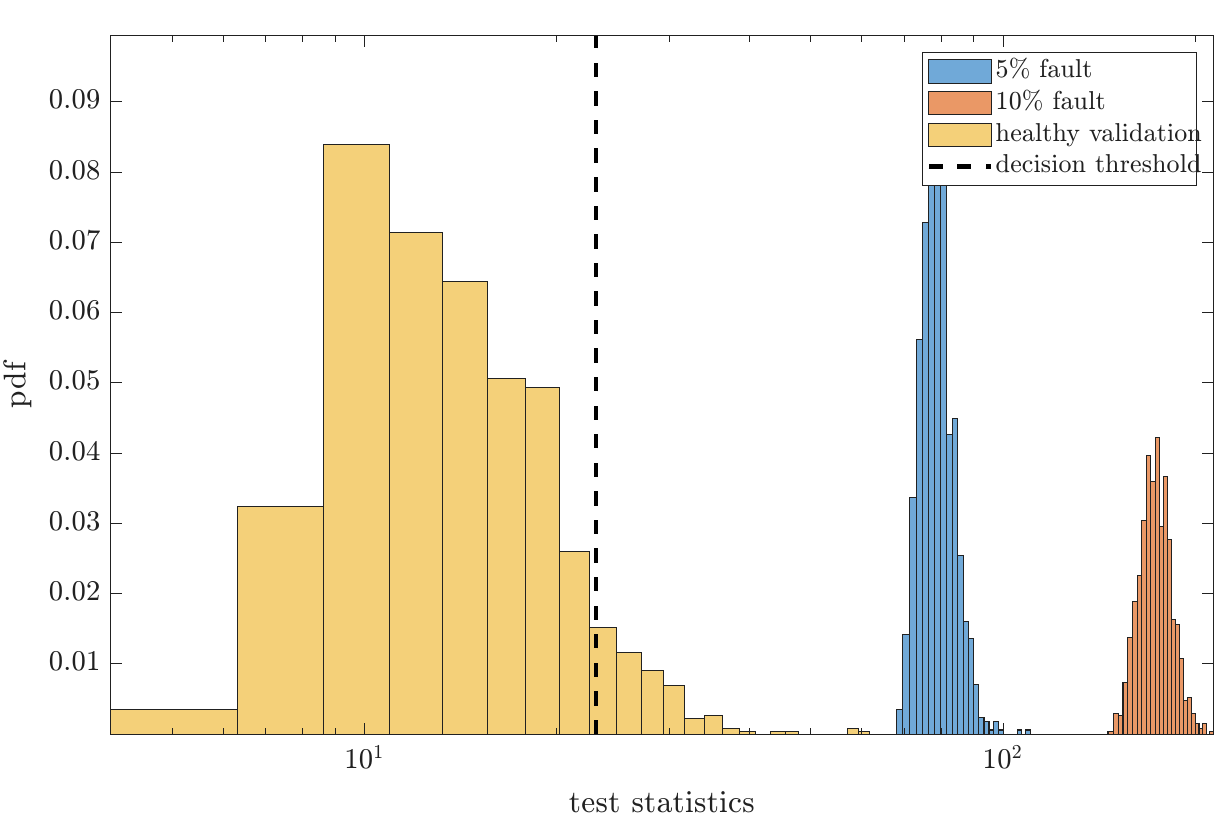}
\end{minipage}
\begin{minipage}[b]{0.49\linewidth}
\center
\includegraphics[width=0.98\textwidth]{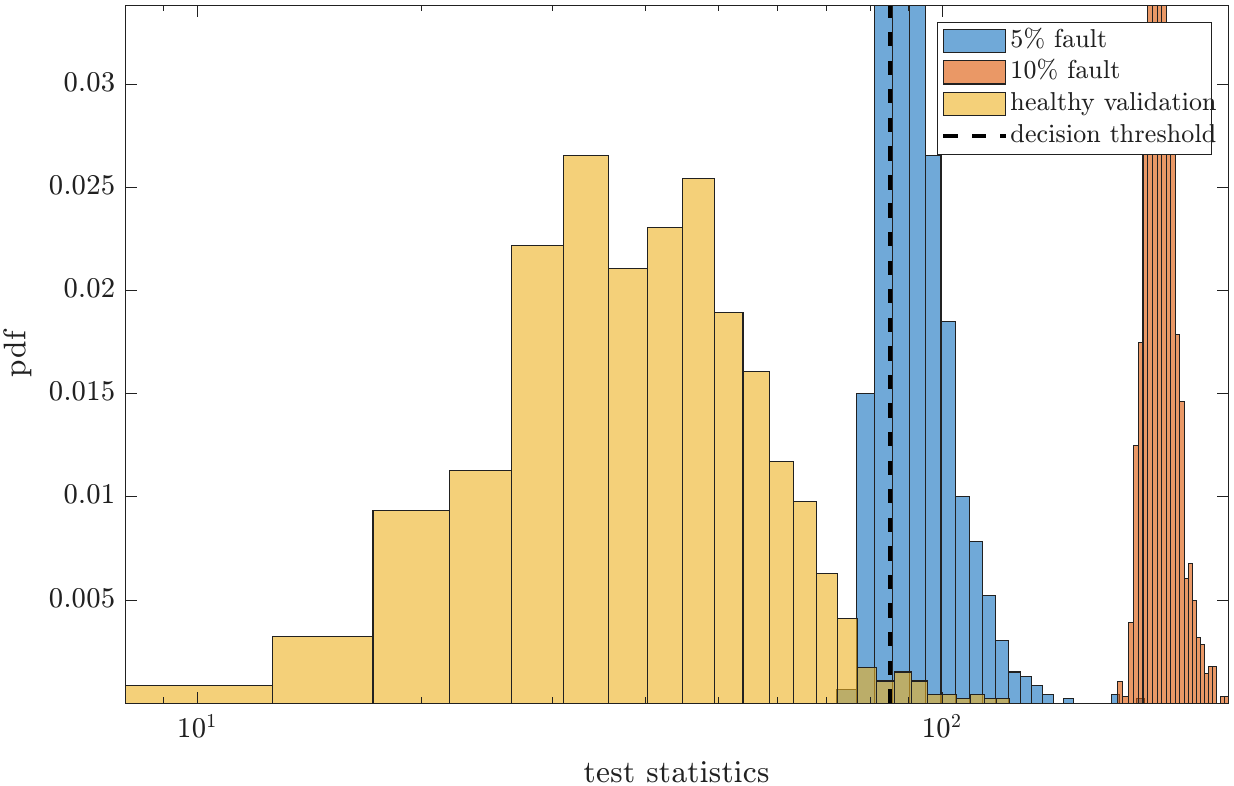}
\end{minipage}
\caption{Distributions of the standard autoencoder test results for a system with cubic nonlinearity. AE 1 (left), VAE 1 (right).} 
\label{fig:results_chainVAE}
\end{figure}

It can be viewed that the results obtained based on learning Hankel matrices of delayed covariances are superior to the ones when using time series directly. While the variational autoencoder and its beta variant offer better false alarm rate in comparison with the proposed approach, the obtained detection rates are much lower. In contrast, a standard autoencoder offers slightly better detection rates than the proposed method, with high false alarm rate as a trade off. Overall, on average, the proposed change detection framework outperforms the applied methods.


\section{Application}\label{sec:application}

In this section the proposed approach is applied on two well-known benchmarks: S101 bridge \cite{michaelS101} and V27 wind turbine blade \cite{AVENDANOVALENCIA2020106686}. The purpose of this application is to empirically illustrate that the proposed method can be adopted to data sets with a realistic number of sensors and its performance on real data is at least as good as the state of art methods in the literature.   

\subsection{S101 bridge}\label{sec:s101}

First, the proposed approach is evaluated using measurement data from the monitoring campaign of the S101 bridge in Austria, shown in Figure~\ref{fig:s101_pic}. The campaign was conducted prior to the demolition of the bridge, during which several controlled damage scenarios were introduced into the structure.  

\begin{figure}[ht!]
\centering
\includegraphics[height=4.05cm]{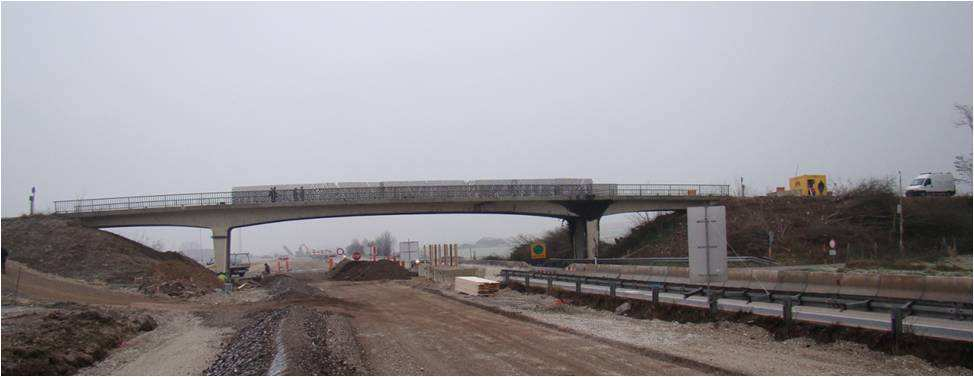}     
\caption{S101 Bridge before demolition.}
\label{fig:s101_pic}
\end{figure}

The dynamic response of the bridge under ambient excitation was recorded using 15 acceleration sensors mounted on the bridge deck. The measurements were sampled at a frequency of $500$\,Hz over a period of three days. During the monitoring campaign, the structure was exposed to environmental variability, as well as artificially introduced damage. Several controlled damage scenarios were considered. A detailed description of the monitoring campaign can be found in~\cite{michaelS101}. The damage scenarios and the corresponding data sets used in this study are summarized in Table~\ref{tab:s101_datasets}.

Data are first low-pass filtered and subsampled to frequency of $25$\,Hz. Out of $150$ healthy data sets available, the first $100$ data sets are used to estimate the reference manifold and the remaining $50$ are used for validation. The covariance Hankel matrices are computed with $p=15$ lags. The autoencoder network is trained with the hidden layer width of $4096$ and a characteristic bottleneck dimension $\ell=20$. The training and autoencoder parameters are chosen similarly to the ones described in Section \ref{sec:validation}.

\begin{table}[ht!]
\centering
\caption{Damage scenarios during the progressive damage test of the S101 Bridge.}
\label{tab:s101_datasets}
\begin{tabular}{clcclc}
\hline
\begin{tabular}[c]{@{}c@{}}Damage\\case 1\end{tabular}
& Damages
& Sets
& \begin{tabular}[c]{@{}c@{}}Damage\\case 2\end{tabular}
& Damages
& Sets \\
\hline
A & First cut, left pier & 5
  & H & 1\textsuperscript{st} tendon cut & 20 \\

B & Second cut, left pier & 15
  & I & 2\textsuperscript{nd} tendon cut & 178 \\

C & Settlement of the left pier, 1\,cm & 10
  & J & 3\textsuperscript{rd} tendon cut & 23 \\

D & Settlement of the left pier, 2\,cm & 21
  & K & 4\textsuperscript{th} tendon cut & 6 \\

E & Settlement of the left pier, 3\,cm & 9
  &   &  &  \\

F & Lifting the left pier +6\,mm & 186
  &   &  &  \\

G & Strengthening the left pier& 45
  &   &  &  \\
\hline
\end{tabular}
\end{table}

Figure \ref{fig:s101results} illustrates that all damage types are detected while the false alarm rate is reasonably low, i.e., at $10\%$. This offers a similar performance as the subspace-based fault detection approach reported in \cite{michaelS101,gres_eurodyn}. 

\begin{figure}[ht!]
\centering
\includegraphics[width=0.75\linewidth]{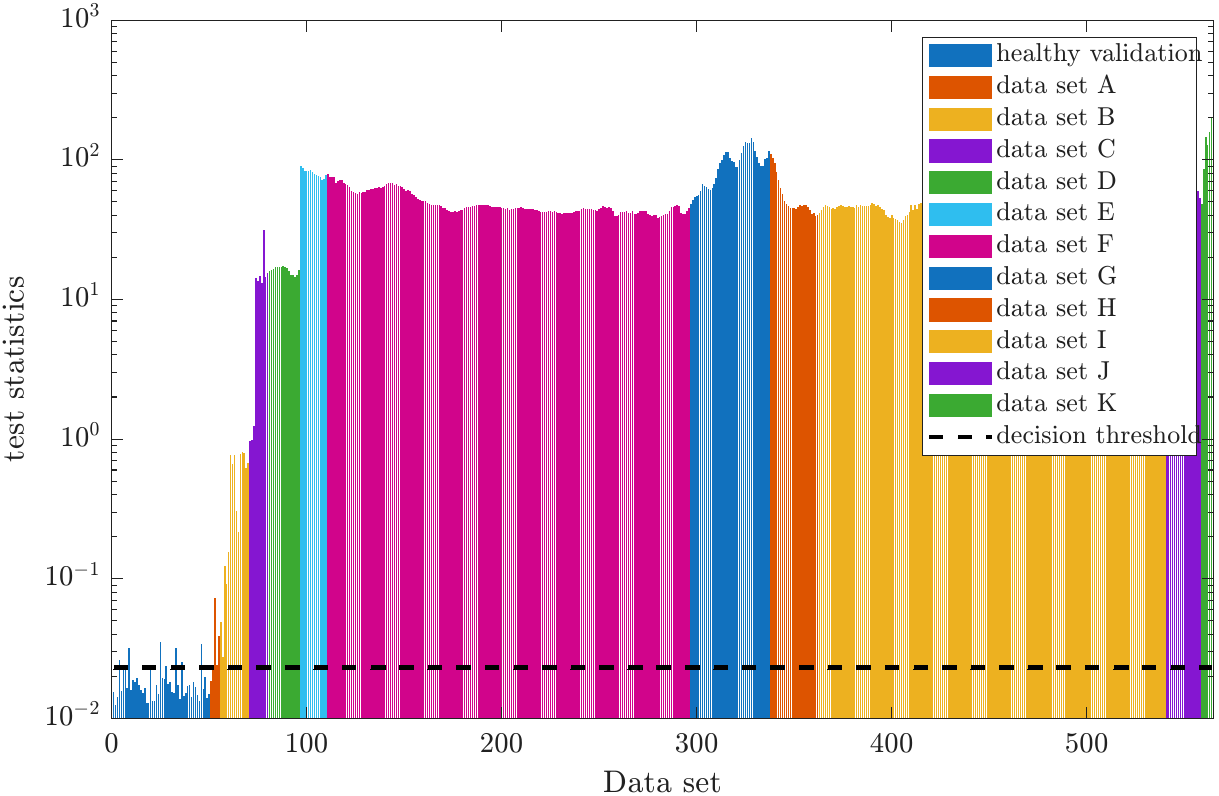}     
\caption{Residual norm obtained for the S101 bridge data sets.}
\label{fig:s101results}
\end{figure}

Subsequently, the proposed method is compared to the autoencoder-based methods benchmarked in Section \ref{sec:validation}. The training procedure follows the same protocol as in the case of the proposed approach and the autoencoders are applied with similar parameters as in Section \ref{sec:validation}. The results are depicted in Table \ref{table:s101_valid}. It can be viewed that the proposed method is outperformed by the standard autoencoder-based anomaly detection when it comes to false alarms, however, for this dataset, it is slightly more sensitive to damage. 

\begin{table}[t!]
\caption{Comparison between the state-of-the-art autoencoder-based anomaly detection methods (AE, VAE, BVAE) and the proposed approach (DAE).}
\vspace{2mm}
\footnotesize
\centering
\begin{tabular}{c*{8}{c}c}
\toprule
{Method} &  {AE 1} &  {AE 2} & {VAE 1} &  {VAE 2} & {BVAE 1} & {BVAE 2} & DAE \\
\toprule
{False alarm rate}           &  {$1\%$}    & {$0\%$}    &{$6.0\%$} &{$0\%$} & {$4.0\%$}& {$0\%$} &  {$10\%$}     \\
{Detection rate}          &  {$97.3\%$}    & {$ 53.1\%$}    &{$74.0\%$} &{$ 51.6\%$} & {$75.5\%$}& {$ 51.6\%$} &  {$99\%$}  \\
\bottomrule
\end{tabular}
\label{table:s101_valid}
\end{table}

\subsection{V27 wind turbine blade}\label{sec:v27}

This section is dedicated to the application of the proposed algorithm on data collected from a wind turbine blade. 
The study is conducted on an experimental benchmark of a Vestas V27 wind turbine, which was analyzed previously for vibration-based damage diagnosis in \cite{UlriksenDDv27_15,AVENDANOVALENCIA2020106686}. In contrast to the previous work, herein the aim is to illustrate that the low frequency data is sensitive to faults and without explicitly accounting for the variable environmental and operating conditions of the wind turbine. 

The blade is 27 m long and is a part of three blade 225kW Vestas wind turbine. 
The blade is subjected to a mix of a random wind excitation and forced impulse inputs, where the latter is induced by an actuator fixed on a surface close to the root of the blade; for more details about the experimental setup see \cite{tcherniak2017active}. 

\begin{figure}[ht!]
\begin{minipage}[b]{0.49\linewidth}
\centering
\includegraphics[width=0.95\linewidth]{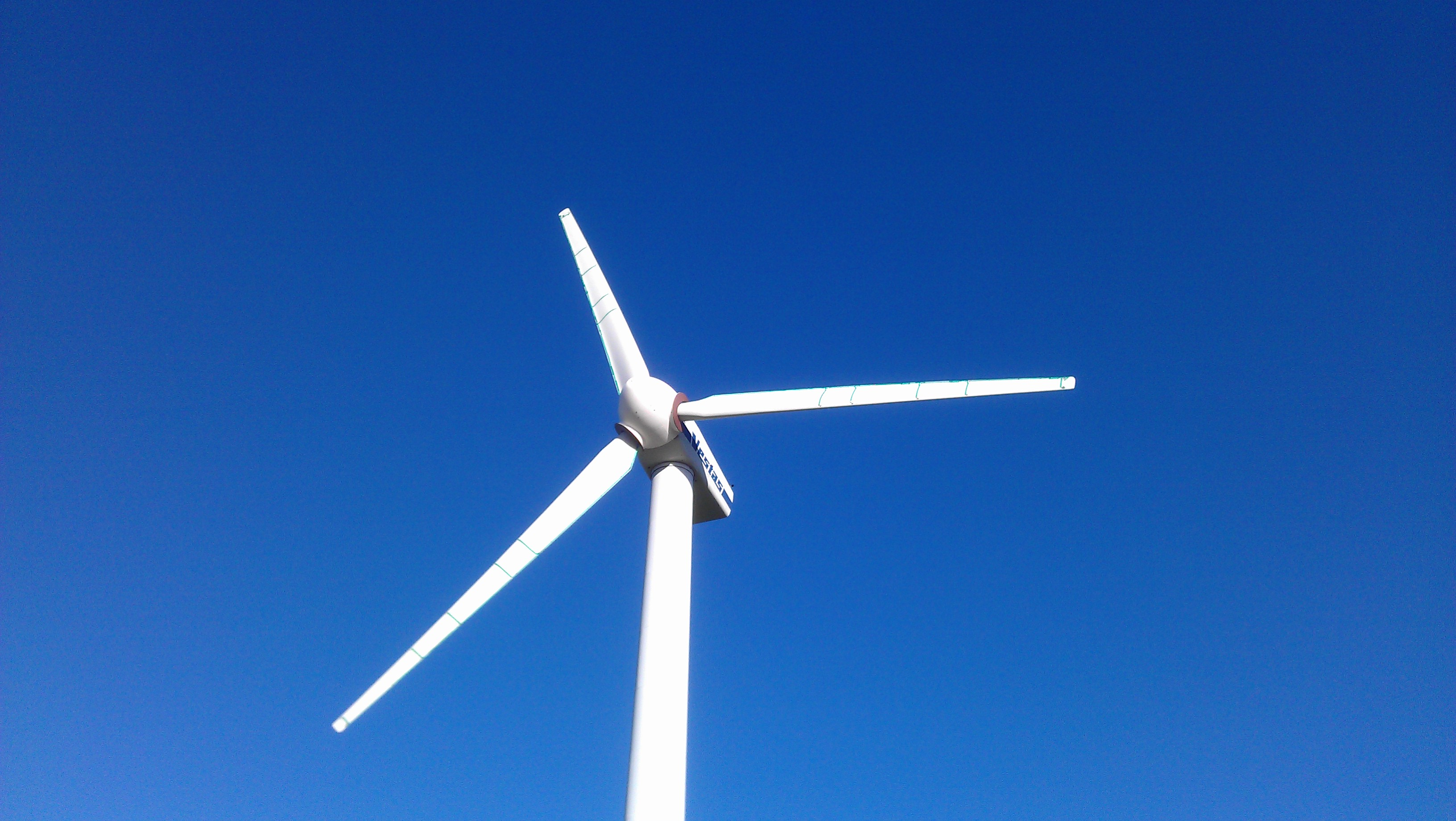}     
\end{minipage}
\begin{minipage}[b]{0.49\linewidth}
\centering
\includegraphics[width=0.95\linewidth]{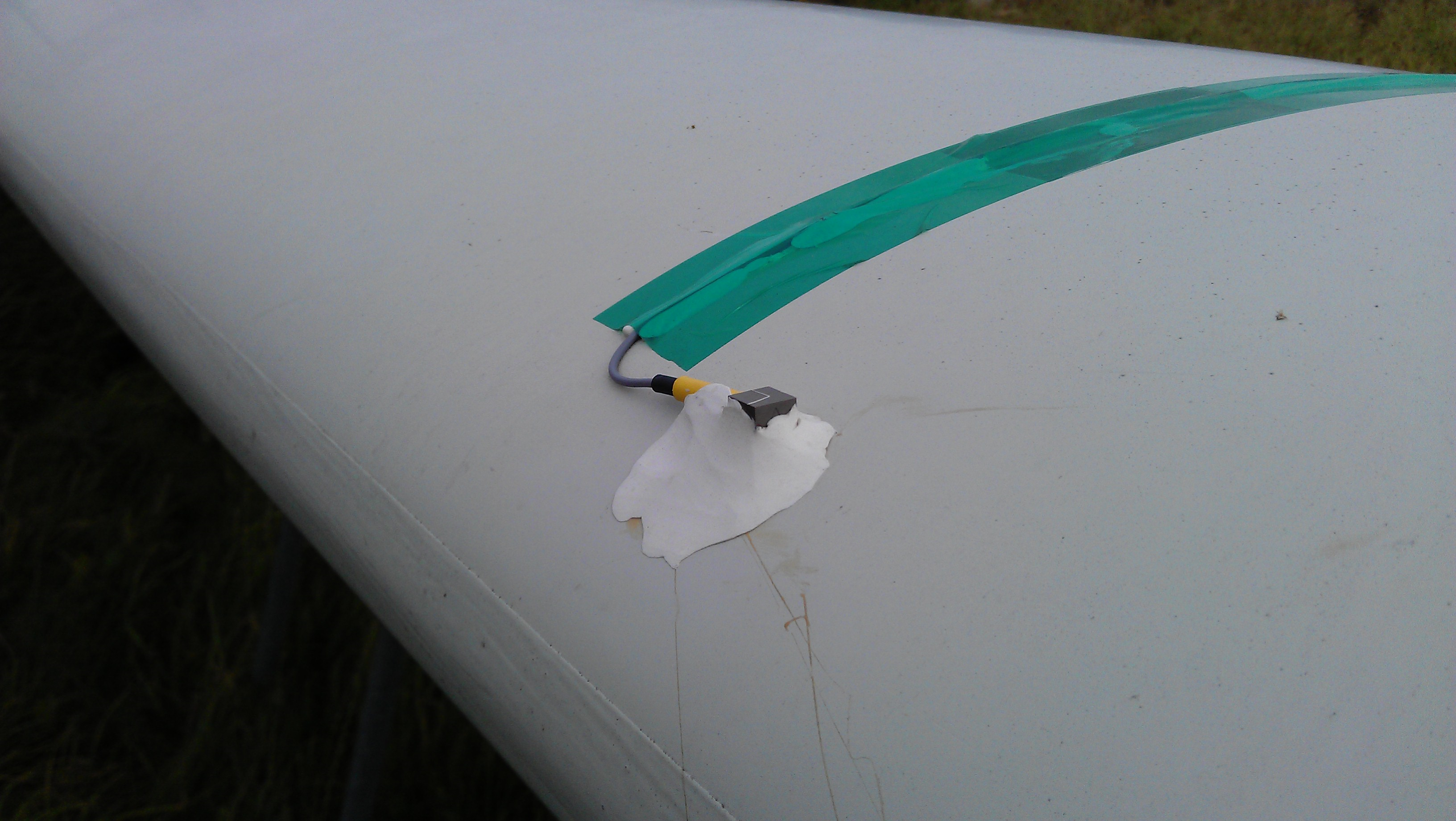} 
\end{minipage}
\caption{Vestas V27 wind turbine (left). V27 blade with an accelerometer attached (right). }
\label{fig:expblade}

\centering
\includegraphics[width=0.95\linewidth]{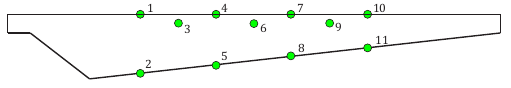}     
\caption{Sketch of V27 blade with sensors. }
\label{fig:labblade}
\end{figure}

The measurements are conducted with a sampling frequency of 16384 Hz and are collected for a period of 3.5 months with 11 accelerometer sensors. Prior to the computation of Hankel matrices data are low-pass filtered and subsampled to $20.825$\,Hz. The experimental setup and the geometry of the blade are shown in Figures \ref{fig:expblade} and \ref{fig:labblade}. During the experiment, artificial damage was progressively introduced by opening the trailing edge of the blade. The initial opening length was $15$\,cm, which was subsequently extended to $30$\,cm and finally to $45$\,cm. Due to the limited computational resources, we have selected 19 measurement sets for the $15$\,cm crack opening, 128 sets for the $30$\,cm opening, and 260 sets for the $45$\,cm opening. The turbine was operated under varying weather conditions during the measurements, as shown in Figure~\ref{fig:Vblades_weather}. The wind turbine also operates in three rotor speed ranges: idle, $32$\,rpm and $42$\,rpm, from which only the data sets from $32$\,rpm and $42$\,rpm are considered. Further details on the wind turbine installation, sensor layout, data processing, and environmental conditions can be found in~\cite{tcherniak2017active,AVENDANOVALENCIA2020106686}.

\begin{figure}[ht!]
\centering
\includegraphics[width=0.65\linewidth]{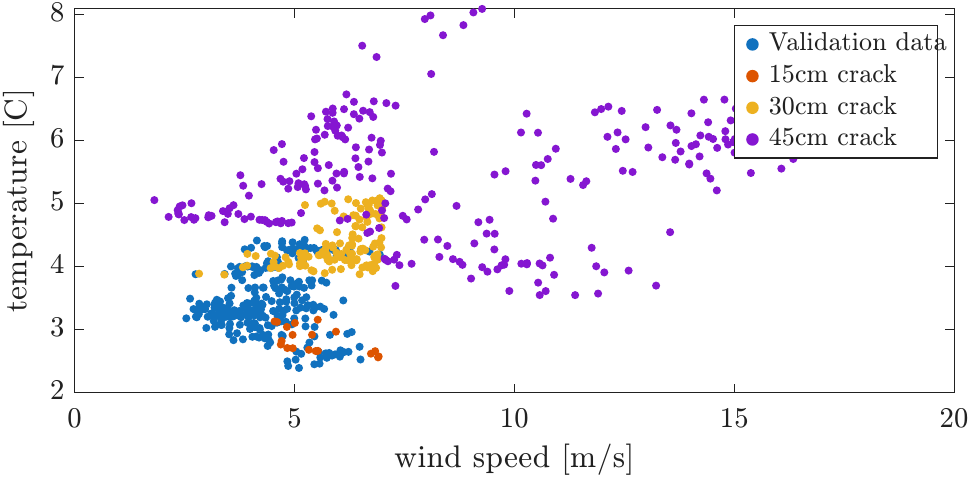}     
\caption{Wind speed and temperature conditions for the selected data sets.}
\label{fig:Vblades_weather}
\end{figure}

In total $398$ healthy data sets were chosen, out of which $108$ data sets were used to estimate the reference manifold and the remaining $290$ were used for validation. The covariance Hankel matrices are computed with $p=25$ lags. The autoencoder network is trained with two hidden layers width of $8192$ and a characteristic bottleneck dimension $\ell=14$. The training and autoencoder parameters are chosen similarly to the ones described in Section \ref{sec:validation}.

The test value corresponding to the analyzed data sets are illustrated in Figure \ref{fig:v27results}. Although each damage scenario is detected, the false alarm rate is $14\%$ and no-detection rate is $8\%$. While the obtained rates are slightly high, the performance of the proposed method on low-frequency data is very competitive to the current results in \cite{AVENDANOVALENCIA2020106686}. The results are expected to improve when the changes in the wind turbine operating parameters and estimation uncertainty are accounted in the design of the fault detection test, which will be considered in our future work.  

\begin{figure}[ht!]
\centering
\includegraphics[width=0.65\linewidth]{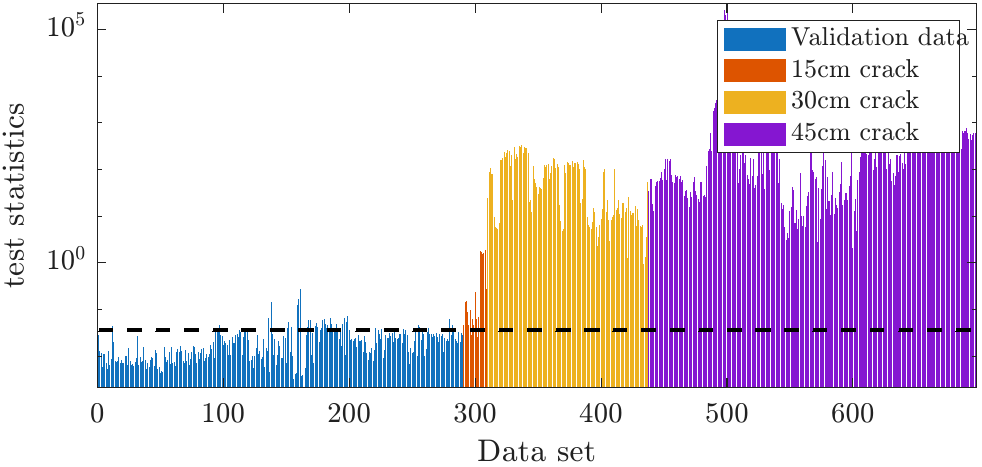}     
\caption{Residual norm obtained for the V27 wind turbine blade data sets.}
\label{fig:v27results}
\end{figure}

Lastly, the proposed method is compared to the autoencoder-based methods used in the previous section. The training procedure follows the same protocol as in the case of the proposed approach and the autoencoders are applied with similar parameters as in Section \ref{sec:validation}. The results are depicted in Table \ref{table:v27_valid}. It can be viewed that the proposed method outperforms the standard autoencoder-based approaches for both false alarm and fault detection rates. 

\begin{table}[t!]
\caption{Comparison between the state-of-the-art autoencoder-based anomaly detection methods and the proposed approach.}
\vspace{2mm}
\footnotesize
\centering
\begin{tabular}{l*{7}{c}}
\toprule
{Metric} & {AE 1} & {AE 2} & {VAE 1} & {VAE 2} & {BVAE 1} & {BVAE 2} & {DAE} \\
\midrule
{False alarm rate} 
& {$30.7\%$} 
& {$0.0\%$} 
& {$23.9\%$} 
& {$0.0\%$} 
& {$28.1\%$} 
& {$0.0\%$} 
& {$14\%$} \\
{Detection rate} 
& {$100.0\%$} 
& {$55.5\%$} 
& {$99.3\%$} 
& {$38.6\%$} 
& {$99.8\%$} 
& {$48.2\%$} 
& {$92\%$} \\
\bottomrule
\end{tabular}
\label{table:v27_valid}
\end{table}


\section{Conclusions}\label{sec:conc}

In this paper a fault detection framework has been developed for nonlinear dynamical systems by interpreting delay output covariance Hankel matrices as samples from a manifold parametrized by stochastic Koopman features. A dual autoencoder was used to learn the nominal manifold and complementary coordinates, whose norm defines a residual for detecting parametric changes. The proposed method was validated on simulations of a nonlinear chain system and applied to the S101 bridge and V27 wind turbine blade benchmarks, where it achieved a favorable balance between false alarm and detection rates compared with standard autoencoder-based approaches. Future work will address the statistical characterization of the residual under finite data, measurement noise and operating variability.

\appendix
\section{Proof of Proposition \ref{prop:hankel_covariance_koopman} } \label{app:proofHcovKop}

Recall that from the stochastic Koopman eigenfunction relation in Section \ref{sec:stochasticKoopman} we have
\[
(\KK\phi_j)(x,\theta) = \E\left[\phi_j(X_{k+1},\theta)\mid X_k=x\right] = \lambda_j(\theta)\phi_j(x,\theta).
\]
Applying this relation recursively for any integer $\ell\geq 1$ gives $\E\left[\phi_j(X_{k+\ell},\theta)\mid X_k=x\right] = \left(\lambda_j(\theta)\right)^\ell\phi_j(x,\theta)$ and stacking the retained eigenfunctions yields $\E\left[\phi(X_{k+\ell},\theta)\mid X_k=x\right] = \Lambda(\theta)^\ell\phi(x,\theta)$. 
When inspecting an individual realization we define a prediction error term $\eta_{k,\ell} \eqdef \phi(x_{k+\ell},\theta) - \Lambda(\theta)^\ell\phi(x_k,\theta)$ such that
\[
\phi(x_{k+\ell},\theta) = \Lambda(\theta)^\ell\phi(x_k,\theta) + \eta_{k,\ell}.
\]
Based on the above, the conditional expectation of the error $\eta_{k,\ell}$ satisfies 
\[
\E\left[\phi(X_{k+\ell},\theta)- \Lambda(\theta)^\ell\phi(X_k,\theta) \mid X_k=x\right] = 0.
\] Therefore, using \eqref{eq:ykstoch}, the measured output realization can be written as
\[
y_{k+\ell} = V(\theta)\Lambda(\theta)^\ell\phi(x_k,\theta) + V(\theta)\eta_{k,\ell} + \varepsilon_n(x_{k+\ell},\theta) + v_{k+\ell},
\]
where subscript $n$ indicates that the residual error corresponds to the truncation to the first $n$ Koopman components.
Subsequently, using Definition \ref{def:datamatrix_def} and \eqref{eq:ykstoch} recall that
$\YY^-=\bar{\YY}^-+\RRR^-+\VV^-$ and $\YY^+=\bar{\YY}^++\RRR^++\VV^+$. Expanding these product gives
\begin{align*}
\YY^+{\YY^-}^{\top} &= \bar{\YY}^+ (\bar{\YY}^-)^{\top} +\bar{\YY}^+ {\RRR^-}^{\top} + \bar{\YY}^+ {\VV^-}^{\top} + \RRR^+ (\bar{\YY}^-)^{\top} + \RRR^+ {\RRR^-}^{\top} + \RRR^+ {\VV^-}^{\top} + \\ &\VV^+ (\bar{\YY}^-)^{\top} + \VV^+ {\RRR^-}^{\top} + \VV^+ {\VV^-}^{\top} .
\end{align*}
In the remainder of this section we will analyze these products one by one. 

First, consider the term $\bar{\YY}^+ (\bar{\YY}^-)^{\top}$. Let $a,b\in\{0,\ldots,p-1\}$ denote the future and past block-row indices. The $(a,b)$ block of $\bar{\YY}^+ (\bar{\YY}^-)^{\top}$ is
\[
\left[\bar{\YY}^+ (\bar{\YY}^-)^{\top}\right]_{a,b} = \frac{1}{N}\sum_{k=0}^{N-1} V(\theta)\phi(x_{k+p+a},\theta) \phi(x_{k+b},\theta)^\H V(\theta)^\H .
\]
Let $\ell=p+a-b$. Since $a,b\in\{0,\ldots,p-1\}$, one has $\ell\geq 1$, and
\[
\phi(x_{k+p+a},\theta)=\phi(x_{k+b+\ell},\theta)=\Lambda(\theta)^\ell\phi(x_{k+b},\theta)+\eta_{k+b,\ell}.
\]
Substitution gives
\begin{align*}
\left[\bar{\YY}^+ (\bar{\YY}^-)^{\top}\right]_{a,b}&=\frac{1}{N}\sum_{k=0}^{N-1} V(\theta) \Lambda(\theta)^\ell \phi(x_{k+b},\theta) \phi(x_{k+b},\theta)^\H V(\theta)^\H
+ \frac{1}{N}\sum_{k=0}^{N-1} V(\theta)\eta_{k+b,\ell} \phi(x_{k+b},\theta)^\H V(\theta)^\H .
\end{align*}
For a fixed $p$, it can be shown that $\frac{1}{N}\sum_{k=0}^{N-1} \phi(x_{k+b},\theta) \phi(x_{k+b},\theta)^\H = G(x,\theta) + o(1)$. 
Furthermore, using the stochastic Koopman relation
\[
\begin{aligned}
&\E\left[\eta_{k+b,\ell}\phi(X_{k+b},\theta)^\H \mid X_{k+b}=x \right] = \left(\E\left[\phi(X_{k+b+\ell},\theta)\mid X_{k+b}=x\right] - \Lambda(\theta)^\ell\phi(x,\theta) \right)\phi(x,\theta)^\H =0 .
\end{aligned}
\]
Hence the sequence $\eta_{k+b,\ell}\phi(X_{k+b},\theta)^\H$ has zero mean and since for a fixed $b$ and $\ell$, it is a finite-lag function of the stationary ergodic Markov chain, it is also stationary and ergodic. Therefore $\frac{1}{N}\sum_{k=0}^{N-1} \eta_{k+b,\ell}\phi(x_{k+b},\theta)^\H = o(1)$. In consequence
\[
\left[\bar{\YY}^+ (\bar{\YY}^-)^{\top}\right]_{a,b}= V(\theta) \Lambda(\theta)^{p+a-b} G(x,\theta) V(\theta)^\H + o(1).
\]
Since $V(\theta) \Lambda(\theta)^{p+a-b}=V(\theta) \Lambda(\theta)^a \Lambda(\theta)^{p-b}$, stacking all block rows and columns gives
\[
\bar{\YY}^+ (\bar{\YY}^-)^{\top}=\OO(\theta)\Gamma(x,\theta)+o(1).
\]

Now, what remains is to bound the terms containing the truncation residual and the measurement noise. Recall that $n$ (the number of Koopman modes) is chosen such that on the region of interest $\|\varepsilon_n(x,\theta)\|\leq \bar{\varepsilon}$, where $\bar{\varepsilon}>0$ is assumed to be small. Then
$\|\RRR^-\|^2 = \frac{1}{N}\sum_{i=0}^{p-1} \sum_{k=0}^{N-1}\|\varepsilon_n(x_{k+i},\theta)\|^2 \leq p\bar{\varepsilon}^2 $ and
$\|\RRR^+\|^2 = \frac{1}{N}\sum_{i=p}^{2p-1} \sum_{k=0}^{N-1} \|\varepsilon_n(x_{k+i},\theta)\|^2 \leq p\bar{\varepsilon}^2 $. 
Thus $\|\RRR^-\|\leq \sqrt{p}\bar{\varepsilon}$ and $\|\RRR^+\|\leq \sqrt{p}\bar{\varepsilon}$.
As such, it can be shown that $\|\bar{\YY}^+{\RRR^-}^{\top}\| \leq \|\bar{\YY}^+\|\|\RRR^-\| \leq \sqrt{p}\bar{\varepsilon}\|\bar{\YY}^+\|$, $\|\RRR^+ (\bar{\YY}^-)^{\top}\| \leq \|\RRR^+\|\|\bar{\YY}^-\|\leq \sqrt{p}\bar{\varepsilon}\|\bar{\YY}^-\|$ and $\|\RRR^+{\RRR^-}^{\top}\| \leq \|\RRR^+\|\|\RRR^-\| \leq p\bar{\varepsilon}^2$.
The mixed terms related to the residual and the noise satisfy $\|\RRR^+{\VV^-}^{\top}\| \leq \|\RRR^+\|\|\VV^-\| \leq \sqrt{p}\bar{\varepsilon}\|\VV^-\|$ and $ \|\VV^+{\RRR^-}^{\top}\| \leq \|\VV^+\|\|\RRR^-\| \leq \sqrt{p}\bar{\varepsilon}\|\VV^+\|$. 

For a fixed $p$, the empirical norms $\|\bar{\YY}^-\|$ and $\|\bar{\YY}^+\|$ are bounded in probability by the bounded fourth-moment assumption, and $\|\VV^+\|$ with $\|\VV^-\|$ are bounded in probability because the measurement noise has finite covariance. Thus all terms containing $\RRR^+$ or $\RRR^-$ are small when $\bar{\varepsilon}$ is small, i.e., the finite Koopman spectrum is a sufficiently good approximation.

Finally, because the measurement noise is zero mean, white, and independent of the state process, and because $p+a-b\geq 1$ for all $a,b\in\{0,\ldots,p-1\}$ the expectations $\E\left[ v_{k+p+a}v_{k+b}^\top \right] = 0$, $\E\left[v_{k+p+a} \phi(X_{k+b},\theta)^\H\right]=0$ and $\E\left[ \phi(X_{k+p+a},\theta)v_{k+b}^\top \right]=0$. As a result $\VV^+{\VV^-}^{\top}=o(1)$, $\VV^+ (\bar{\YY}^-)^{\top}=o(1)$ and $\bar{\YY}^+ {\VV^-}^{\top}=o(1)$. Combining these relations gives
\[
\YY^+{\YY^-}^{\top} = \OO(\theta)\Gamma(x,\theta) + O(\bar{\varepsilon}) + O(\bar{\varepsilon}^2) + o(1).
\]
Thus, for a sufficiently small truncation residual and sufficiently large $N$
\[
\est{\HH}(x,\theta) \approx \est{\OO}(\theta)\est{\Gamma}(x,\theta),
\]
where $\est{\OO}(\theta)$ and $\est{\Gamma}(x,\theta)$ denote the estimates of the exact quantities.

\section{Proof of Proposition~\ref{prop:hankel_covariance_manifold}}
\label{app:proof_hankel_covariance_manifold}

For this proof, $x\in\mathcal X_N$ denotes the finite vector obtained by stacking all states required to construct the past and future data matrices in Definition~\ref{def:datamatrix_def}. Since the retained stochastic Koopman eigenfunctions are smooth, every entry of
\[
G(x,\theta)=\frac{1}{N}\sum_{k=0}^{N-1}\varphi(x_k,\theta)\varphi(x_k,\theta)^H
\]
is a finite sum of products of smooth functions and therefore is also smooth with respect to $x$.

For a fixed $\theta$, the matrices $V(\theta)$, $\Lambda(\theta)$, and $B(\theta)$ do not depend on $x$. It follows from the definitions of $\Gamma(x,\theta)$ and $g(x,\theta)$ that these quantities are smooth with respect to $x$. Consequently, \eqref{eq:HankelVec} shows that $x\mapsto\est{h}(x,\theta)$ is a smooth map.
Then, for every $x_0\in\mathcal X_N$, the image of a sufficiently small neighborhood of $x_0$ under this map is a smooth embedded manifold in $\mathbb R^{(pr)^2}$. Its dimension is equal to the constant rank of the map in that neighborhood. Therefore, vectorized covariance estimates obtained from \eqref{eq:HankelVec} from finite realizations within the same local region may be interpreted as samples from the same local manifold.

\section{Description of the chain system} \label{app:chainapp}

\subsection{Periodic stiffness}

The vibration behavior of nonlinear mechanical system with periodic stiffness coefficients can be described by the differential equation
\beq\label{eq:modelznl}
\MM \ddot{z}(t) + \CC \dot{z}(t) + \KK(z(t)) z(t) = f(t) .
\enq 
The nonlinear stiffness matrix $\KK(z(t))$ is defined as
\begin{align*}
\KK(z(t)) \eqdef \vect{\bar{k}_1 + \bar{k}_2 & -\bar{k}_2 & 0 & 0 &  0 & 0 \\ -\bar{k}_2 & \bar{k}_2+\bar{k}_3 & -\bar{k}_3 & 0 &  0 & 0 \\ 0 & -\bar{k}_3 & \bar{k}_3+\bar{k}_4 & -\bar{k}_4 &  0 & 0 \\ 0 & 0 & -\bar{k}_4 & \bar{k}_4+\bar{k}_5 & -\bar{k}_5 & 0 \\ 0 & 0 & 0 & -\bar{k}_5 & k_5+\bar{k}_6 & -\bar{k}_6 \\ 0 & 0 & 0 & 0 & -\bar{k}_6 & \bar{k}_6}
\end{align*}
where for $i=1\hdots6$ $\bar{k}_i = k_i + a\cos( b z_i(t))$, where $a = 10$, $b = 50$ and $z_i(t)$ denotes the $i$'th entry of the displacement vector. System \eqref{eq:modelznl} can be transformed into the discrete-time model \eqref{eq:statenl}--\eqref{eq:outputnl} using numerical integration methods, e.g., fourth-order Runge Kutta.   

\subsection{Duffing spring stiffness}

The vibration behavior of the nonlinear mechanical system with Duffing springs can be described by
\beq\label{eq-modelznl}
\MM \ddot{z}(t) + \CC \dot{z}(t) + \KK z(t) + g(z(t)) = f(t) ,
\enq
where $g(z(t))$ is a nonlinear restoring force, which is shown to be a cubic function of displacements below. Let the restoring force in each spring be
\[
s_i(t) = k_i \delta_i(t) + \gamma_i \delta_i^3(t),
\]
where $k_i$ is the linear stiffness, $\gamma_i$ is the cubic Duffing stiffness coefficient and $\delta_i(t) = z_i(t)-z_{i-1}(t)$ denotes the elongation of the $i$th spring. The nonlinear contribution of spring $i$ is thus
\[
s_i^{\mathrm{nl}}(t)=\gamma_i \delta_i^3(t).
\]
The nonlinear force vector is then obtained by assembling all cubic spring forces as
\[
g(z(t)) =
\vect{
s_1^{\mathrm{nl}}(t)-s_2^{\mathrm{nl}}(t)\\
s_2^{\mathrm{nl}}(t)-s_3^{\mathrm{nl}}(t)\\
s_3^{\mathrm{nl}}(t)-s_4^{\mathrm{nl}}(t)\\
s_4^{\mathrm{nl}}(t)-s_5^{\mathrm{nl}}(t)\\
s_5^{\mathrm{nl}}(t)-s_6^{\mathrm{nl}}(t)\\
s_6^{\mathrm{nl}}(t)
}, 
\]
which boils down to
\[
g(z(t)) =
\vect{
\gamma_1 z_1^3(t)-\gamma_2\big(z_2(t)-z_1(t)\big)^3\\
\gamma_2\big(z_2(t)-z_1(t)\big)^3-\gamma_3\big(z_3(t)-z_2(t)\big)^3\\
\gamma_3\big(z_3(t)-z_2(t)\big)^3-\gamma_4\big(z_4(t)-z_3(t)\big)^3\\
\gamma_4\big(z_4(t)-z_3(t)\big)^3-\gamma_5\big(z_5(t)-z_4(t)\big)^3\\
\gamma_5\big(z_5(t)-z_4(t)\big)^3-\gamma_6\big(z_6(t)-z_5(t)\big)^3\\
\gamma_6\big(z_6(t)-z_5(t)\big)^3
}.
\]

For $\gamma_i>0$ the corresponding spring exhibits hardening behavior, whereas for $\gamma_i<0$ it exhibits softening behavior. For the purpose of this work $\gamma_i = 50$. System \eqref{eq-modelznl} can be transformed into the discrete-time model \eqref{eq:statenl}--\eqref{eq:outputnl} using numerical integration methods, e.g., fourth-order Runge--Kutta.

\bibliographystyle{elsarticle-num}
\bibliography{bibl_gen}

\end{document}